\documentclass[journal,twoside,web,xcdraw,svgnames]{ieeecolor}
\usepackage{generic}
\usepackage{cite}
\usepackage{amsmath,amssymb,amsfonts}

\usepackage{hyperref}

\usepackage[T1]{fontenc}
\usepackage[utf8]{inputenc}
\usepackage{graphicx}
\usepackage{amsmath,systeme,amssymb,amsfonts}
\usepackage[version=4]{mhchem}
\usepackage{siunitx}
\usepackage{longtable,tabularx}
\usepackage{comment}
\usepackage{float}
\usepackage[utf8]{inputenc}
\usepackage{tikz, pgfplots}
\usepackage{textcomp}
\usepackage{tcolorbox} % for colored boxes
\usepackage{graphicx}
\usepackage{amsmath}

\usepackage[version=4]{mhchem}
\usepackage{siunitx}
\usepackage{longtable,tabularx}
\usepackage{tikz}
\usetikzlibrary{positioning, shapes, arrows}
\usepackage{thmtools}
\usepackage{textcomp}
\usepackage{soul}
\usepackage{float,amsfonts,amsthm,color}
\usepackage{amssymb}
\usepackage{graphics} % for pdf, bitmapped graphics files
\usepackage{booktabs} % For formal tables
\usepackage{caption, subcaption}
\usepackage{breqn}
\usepackage{array, mathtools}
\usepackage{pgfplots}
\usepackage{siunitx}
\usepackage{algorithm}
\usepackage[noend]{algpseudocode}
\usepackage{tabulary}
\usepackage{stmaryrd} % for double brackets
\usepackage{varwidth}
\usepackage{comment}
\usetikzlibrary{positioning}

\newtheorem{theorem}{Theorem}
\newtheorem{definition}{Definition}

\newtheorem{proposition}{Proposition}
\newtheorem{lemma}{Lemma}
\newtheorem{remark}{Remark}
\newtheorem{corollary}{Corollary}
\newtheorem{example}{Example}

\def\BibTeX{{\rm B\kern-.05em{\sc i\kern-.025em b}\kern-.08em
    T\kern-.1667em\lower.7ex\hbox{E}\kern-.125emX}}
\begin{document}
\title{Can a Learner Regret Using a No–Regret Algorithm? A Control–Theoretic Study of Performance Dominance}
\author{Hassan Abdelraouf, \IEEEmembership{Member, IEEE}, and Jeff S. Shamma, \IEEEmembership{Fellow, IEEE}
\thanks{Hassan Abdelraouf is with the Department of Aerospace Engineering, University of Illinois at Urbana-Champaign, USA. (e-mail: hassana4@illinois.edu) }
\thanks{Jeff S. Shamma is with the Department of Industrial and Enterprise Systems Engineering, University of Illinois at Urbana-Champaign, USA. (e-mail: jshamma@illinois.edu) }}

\maketitle

\begin{abstract}
No-regret learning dynamics ensure that a learner asymptotically achieves an average reward no worse than that of any fixed strategy. This no-regret guarantee does not determine the value of the asymptotic average reward. Indeed, it is possible for different no-regret learning dynamics to exhibit different asymptotic average rewards when facing the same environment while both assure the no-regret guarantee. This paper asks whether a ``free-lunch'' phenomenon can arise among no-regret algorithms. Namely, is it possible for one no-regret learning rule to uniformly outperform another no-regret learning rule across all payoff environments. Stated differently, can a learner regret not using a particular no-regret algorithm? We consider generalized replicator dynamics (RD) as a cascade interconnection between a linear time-invariant (LTI) system and the softmax nonlinearity. Varying this LTI system leads to different realizations of replicator dynamics, including so-called anticipatory RD, exponential RD, and other forms of higher-order RD. Setting the LTI system to be an integrator realizes standard RD, which is known to satisfy the no-regret property. Within this framework, we analyze and compare various realizations of these generalized realizations RD by varying the LTI system. We first formulate performance comparison as a passivity property of an associated comparison system and establish ``local’’ dominance results, i.e., comparing the asymptotic performance near an equilibrium payoff vector. We then cast performance comparison between a form of anticipatory RD and standard RD as an optimal-control problem. We show that the minimal achievable cumulative reward gap is zero, thereby establishing global dominance of anticipatory RD across all payoff environments and establishing a ``free lunch’’ among no-regret learning dynamics.

\end{abstract}

%%%%%%%%%%%%%%%%%%%%%%%%%%%%%%%%%%%%%%%%%%%%%%%%%%%%
\section{Introduction}
\label{sec:introduction}
As autonomous AI agents become more widely deployed across dynamic, multi-agent environments, they will continuously learn and interact in real time to achieve complex goals. In this setting, online learning has emerged as a powerful framework for adaptive decision-making in uncertain, competitive environments. A learning agent selects actions without a full model, observes feedback, and updates its strategy  in real time. This approach is widely used in online recommendation systems \cite{hazan2016introduction}, traffic routing and congestion management \cite{gollapudi2023online}, network resource allocation \cite{chen2017online}, and market prediction \cite{lesage2020predictive}. Performance is typically measured by regret, the difference between the learner’s cumulative utility and that of the best fixed strategy in hindsight. An algorithm is said to be \emph{no-regret} if its regret grows sublinearly in time, so that the average regret vanishes asymptotically \cite{hazan2007logarithmic,shalev2012online,zhang2013online}. A broad family of first-order methods has this guarantee, including online mirror descent (OMD) \cite{shalev2006convex,yuan2018optimal}, follow-the-regularized-leader (FTRL) \cite{shalev2007primal}, follow-the-perturbed-leader (FTPL)  \cite{kalai2005efficient}, multiplicative weight updates (MWU) \cite{freund1997decision,arora2012multiplicative}, EXP3 \cite{auer2002nonstochastic} , and online gradient descent (OGD) \cite{zinkevich2003online}.

When the learner interacts with other adaptive learners, each agent’s utility depends not only on its own action but also on the actions of others. \emph{Game theory} provides the mathematical framework to model and analyze such strategic interactions \cite{roughgarden2010algorithmic,myerson2013game}. In this context, learning in games refers to adaptive procedures in which each agent updates its strategy based on past observations (realized payoffs and opponents’ actions) and then plays the updated strategy in the next round \cite{fudenberg1998theory,young2004strategic}. A foundational result in this literature is that if every player runs a \emph{no-regret} algorithm, then the empirical distribution of joint play converges to the set of \emph{coarse correlated equilibria} (CCE); equivalently, the time-averaged play approaches CCE without centralized coordination \cite{foster1997calibrated,hart2000simple,cesa2006prediction}.
These guarantees help explain the success of no-regret learning in deployed multi-agent systems, including distributed sensor coverage and target assignment \cite{arslan2004distributed}, multi-robot coordination and consensus \cite{marden2009cooperative}, and distributed energy and resource management \cite{marden2013model}. 

To move from equilibrium outcome to the process by which adaptive learners reach them, it is natural to adopt a dynamical-systems perspective. Within this perspective, passivity is a fundamental input--output property that captures energy conservation and dissipation in mechanical and electrical systems \cite{willems1972dissipative}. In the online learning setting, a learning algorithm can be viewed as an input--output operator mapping payoff signals to strategy trajectories. This perspective enables the use of passivity--based
tools for both analysis and design.

Building on this perspective, \cite{fox2013population} introduced the notion of $\delta$--passivity to analyze convergence of learning dynamics to Nash equilibria in contractive games, modeling learning in games as a feedback interconnection between a \emph{learning operator} (payoff $\to$ strategy) and a \emph{game operator} (strategy $\to$ payoff). Subsequent works \cite{gao2020passivity,gao2023second} employed equilibrium--independent passivity to study convergence in zero--sum games and to design higher--order learning dynamics that converge in settings where first--order models fail. More recently, cite{abdelraouf2025passivity} proposed a passivity--based classification of learning dynamics according to the notion of passivity they satisfy, incremental passivity, equilibrium--independent passivity, and $\delta$-passivity, and used it as a
unifying framework to analyze convergence in contractive games. Moreover, \cite{martins2025counterclockwise} employs counterclockwise
dissipativity to study convergence in potential games.

Going beyond convergence guarantees, the primary objective of a learning agent is to maximize its cumulative reward over time. Finite regret ensures that, in the long run, a learning dynamics achieves an average reward no worse than that of any fixed strategy in hindsight. Building on this perspective, \cite{abdelraouf2025passivity} establishes a rigorous equivalence between finite regret and a passivity property: passivity from the payoff signal to the
deviation of the evolving strategy from any fixed strategy in the simplex is equivalent to finite regret.

This result raises a fundamental question: is finite regret alone sufficient to assess the performance of learning dynamics as measured by asymptotic average reward? In particular, when two learning models both satisfy no--regret guarantees, does one  outperform the other in the long run? Equivalently, does there exist a ``free lunch'' among no--regret algorithms, namely, a learning rule that achieves uniformly higher average reward across all payoff environments? If so, a learner may incur regret not for violating no--regret, but for selecting an suboptimal no--regret algorithm. Addressing this question is the central focus of this paper. The main contributions are summarized as follows:

\begin{itemize}
    \item We provide explicit counterexamples demonstrating that finite regret
alone is not sufficient to assess the performance of learning dynamics, by
showing environments in which finite--regret dynamics outperform dynamics
without finite regret, and vice versa (Section~\ref{sec:motivation}).

   \item We model payoff--based higher--order replicator dynamics as a cascade
    interconnection between a diagonal LTI system $G(s)=g(s)I_n$ and the softmax
    mapping, and show that anticipatory replicator dynamics is dynamically
    equivalent to predictive replicator dynamics equipped with a first--order
    low--pass predictor (Section~\ref{sec:pay_off_based_higher_order_RD}).
    
   \item We introduce an oracle replicator dynamics (predictive RD with perfect
    prediction) and prove that it uniformly dominates standard replicator
    dynamics for all payoff trajectories (Section~\ref{sec:oracle_RD}).

   \item We provide a frequency--domain interpretation of performance for
higher--order replicator dynamics, modeled as a cascade interconnection between
a diagonal LTI system $g(s)I_n$ and the softmax mapping, and establish rigorous
links between the asymptotic average reward of the learning dynamics and the
frequency response $g(j\omega)$. In sinusoidal environments, we show that:
(i) any learning dynamics with passive $g(s)$ outperforms standard replicator
dynamics; (ii) for two learning models with $g_1(j\omega)$ and $g_2(j\omega)$, if \(|g_1(j\omega)|=|g_2(j\omega)|\), the one with higher \(\cos(\arg(g(j\omega)))\) has better asymptotic average reward; and (iii) if $g_1(j\omega)$ and $g_2(j\omega)$ have the same phase, the one with higher magnitude has better performance.  (Section \ref{sec:frequency_domain_intuition}). 
    
   \item We show that the predictive exponential replicator dynamics equipped with a
first--order low--pass predictor uniformly dominates the standard exponential
replicator dynamics for arbitrary payoff trajectories
(Section~\ref{sec:predictive_Ex_RD_and_Ex_RD}).

   \item We Formulate the performance comparison between two learning models as a passivity problem, i.e., the passivity of the associated comparison dynamical system is equivalent to the uniform dominance of one model over the other. Using this framework, we show that predictive RD with any passive asymptotically stable predictor, including the anticipatory replicator dynamics, locally dominates the standard replicator dynamics (Subsection \ref{subsec:performance_comparison_local}). 
  
   \item We cast the global comparison between the anticipatory replicator dynamics  and the standard replicator dynamics as an optimal--control problem and show that the minimal
    cumulative reward gap is zero, thereby establishing global uniform dominance
    of anticipatory replicator dynamics over standard replicator dynamics for
    arbitrary payoff trajectories. This constitutes the main ``free lunch''
    result of the paper (Subsection~\ref{subsec:performance_comparison_global}).
   
\end{itemize}

%%%%%%%%%%%%%%%%%%%%%%%%%%%%%%%%%%%%%%%%%%%%%%%%%%%
\section{Preliminaries}
%%%%%%% Notations %%%%%%%%%%%%%%%%%%%%
\subsection{Notations}
Let $\mathbb{R}_+$ denote the set of nonnegative real numbers.
For a vector $x\in\mathbb{R}^n$, $x_i$ denotes its $i$-th component and
$[x]_+$ denotes its componentwise nonnegative part, i.e.,
$([x]_+)_i=\max\{x_i,0\}$.
The Euclidean inner product and norm are denoted by
$\langle x,y\rangle=x^\top y$ and $\|x\|_2$, respectively.
The matrix $\operatorname{diag}(x)$ denotes the diagonal matrix with diagonal
entries $x_1,\dots,x_n$. We write $I_n$ for the $n\times n$ identity matrix,
$\mathbf{1}_n$ for the vector of all ones, and $\mathbf{0}_n$ for the zero vector. The probability simplex in $\mathbb{R}^n$ is
\[
\Delta_n \;\coloneqq\; \bigl\{ s\in\mathbb{R}^n \;:\; s_i\ge 0 \ \forall i,\ \mathbf{1}_n^\top s = 1 \bigr\},
\]
and $\operatorname{Int}(\Delta_n)$ denotes its interior.
The vector $e_i\in\Delta_n$ denotes the $i$-th vertex of the simplex, i.e.,
$(e_i)_i=1$ and $(e_i)_j=0$ for $j\neq i$.
For a closed convex set $C\subset\mathbb{R}^n$, $\Pi_C(x)$ denotes the Euclidean
projection of $x$ onto $C$,
\[
\Pi_C(x) \;\coloneqq\; \arg\min_{s\in C} \|x-s\|_2 .
\]
The softmax mapping $\sigma:\mathbb{R}^n\to\operatorname{Int}(\Delta_n)$ is defined as
\[
\sigma(x) \;\coloneqq\; \frac{\exp(x)}{\sum_{i=1}^n \exp(x_i)},
\]
and the log-sum-exp function $\operatorname{lse}:\mathbb{R}^n\to\mathbb{R}$ is given by
\[
\operatorname{lse}(x) \;\coloneqq\; \log\!\Bigl(\sum_{i=1}^n e^{x_i}\Bigr).
\]
For probability vectors $x,y\in\Delta_n$, the Kullback--Leibler (KL) divergence
is defined as
\[
D_{\mathrm{KL}}(x\|y)
\;\coloneqq\;
\sum_{i=1}^n x_i \ln\!\left(\frac{x_i}{y_i}\right).
\]

The space of square-integrable functions on finite intervals is denoted by
$\mathbb{L}_e$, i.e.,
\[
\mathbb{L}_e \;\coloneqq\;
\Bigl\{ f:\mathbb{R}_+\to\mathbb{R}^n \;:\;
\int_0^T f(t)^\top f(t)\,dt < \infty \;\; \forall T>0 \Bigr\}.
\]
For $f,g\in\mathbb{L}_e$, the truncated inner product over $[0,T]$ is
\[
\langle f,g\rangle_T \;\coloneqq\; \int_0^T f(t)^\top g(t)\,dt,
\]
and the induced $\mathbb{L}_2$-norm on $[0,T]$ is
\[
\|f\|_{2,[0,T]}
\;\coloneqq\;
\Bigl(\int_0^T \|f(t)\|_2^2\,dt\Bigr)^{1/2}.
\]
For a stable transfer function $g(s)$, $\|g\|_\infty$ denotes its
$\mathcal{H}_\infty$ norm,
\[
\|g\|_\infty \;\coloneqq\; \sup_{\omega\in\mathbb{R}} |g(j\omega)|.
\]
%%%%%%%%%%% Online Learning %%%%%%%%%%%%%%%%%%%%%%%%%%%%%%%
\subsection{Online Learning}
We consider continuous–time online learning dynamics in which the payoff vector
$p(\cdot):\mathbb{R}_+\to\mathbb{R}^n$ is assumed to be continuously differentiable.
Our primary focus is on the \emph{replicator dynamics} (RD)
\cite{mertikopoulos2010emergence,hofbauer2009time}, which is the
continuous–time counterpart of multiplicative–weights algorithms
\cite{cheung2018multiplicative,palaiopanos2017multiplicative}.

In RD, the learning process starts  at time $t=0$ with an initial score vector $z(0)\in\mathbb{R}^n$. For $t>0$, the learner accumulates payoffs according to
\(
z(t)=z(0)+\int_0^t p(\tau)\,d\tau,
\)
then selects a mixed strategy $x(t)\in\Delta_n$ via the softmax mapping
\(
x(t)=\sigma\bigl(z(t)\bigr).
\)
Equivalently, the learning process is governed by the dynamical system
\begin{equation}
\label{eq:RD}
\begin{aligned}
\dot z &= p,\\
x &= \sigma(z).
\end{aligned}
\end{equation}
This formulation has been extensively studied in evolutionary game theory and
online learning; see, e.g., \cite{mertikopoulos2016learning,mertikopoulos2017convergence}.

From an input–output perspective, RD can be viewed as an operator mapping the
payoff signal $p(\cdot)$ to the strategy trajectory $x(\cdot)$. In particular,
\eqref{eq:RD} can be represented as a cascade interconnection
between the integrator $\tfrac{1}{s}I_n$ and the static nonlinearity $\sigma(\cdot)$,
as shown in Fig.~\ref{fig:rd_block}.

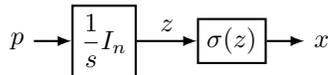
\begin{figure}[h]
\centering
\begin{tikzpicture}[>=latex, thick, node distance=2.5cm]
\tikzstyle{block} = [draw, rectangle, minimum height=0.5cm, minimum width=0.8cm, align=center]
\node (p) {$p$};
\node [block, right=0.5cm of p] (G) {$\dfrac{1}{s}I_n$};
\node [block, right=0.8cm of G] (C) {$\sigma(z)$};
\node [right=0.5cm of C] (x) {$x$};
\draw[->] (p) -- (G);
\draw[->] (G) -- node[above] {$z$} (C);
\draw[->] (C) -- (x);
\end{tikzpicture}
\caption{Block–diagram representation of replicator dynamics.}
\label{fig:rd_block}
\end{figure}

Several alternative continuous–time learning dynamics have been proposed in the
literature. Notable examples include:

\begin{enumerate}
\item \emph{Brown–von Neumann–Nash (BNN) dynamics} \cite{brown1950solutions}:
\begin{equation}
\dot{x}_i = [p_i-x^\top p]_+ - x_i \sum_{j=1}^n [p_j-x^\top p]_+ .
\label{eq:BNN}
\end{equation}

\item \emph{Smith dynamics} \cite{smith1984stability}:
\begin{equation}
\dot{x}_i = \sum_{j=1}^n x_j [p_i-p_j]_+ - x_i \sum_{j=1}^n [p_j-p_i]_+ .
\label{eq:Smith}
\end{equation}

\item \emph{Target projection (TP) dynamics} \cite{tsakas2009target}:
\begin{equation}
\dot{x} = -x+\Pi_{\Delta_n}(x+p).
\label{eq:TP}
\end{equation}

\item \emph{Exponential replicator dynamics (Ex–RD)} \cite{gao2020passivity}:
\begin{equation}
\label{eq:EX-RD}
\begin{aligned}
\dot z &= \lambda(p-z),\\
x &= \sigma(z),
\end{aligned}
\end{equation}
where $\lambda>0$ is a design parameter.
\end{enumerate}

 Compared to standard RD, the score dynamics in \eqref{eq:EX-RD} incorporate an exponential forgetting mechanism, yielding a first–order low–pass filter between the payoff and the score. Accordingly, Ex–RD has a block–diagram representation as the cascade interconnection of the stable transfer function $\tfrac{\lambda}{s+\lambda}I_n$ with the softmax mapping, as shown in Fig.~\ref{fig:exrd_block}.

\begin{figure}[h]
\centering
\begin{tikzpicture}[>=latex, thick, node distance=2.5cm]
\tikzstyle{block} = [draw, rectangle, minimum height=0.5cm, minimum width=0.8cm, align=center]
\node (p) {$p$};
\node [block, right=0.5cm of p] (G) {$\dfrac{\lambda}{s+\lambda}I_n$};
\node [block, right=0.8cm of G] (C) {$\sigma(z)$};
\node [right=0.5cm of C] (x) {$x$};
\draw[->] (p) -- (G);
\draw[->] (G) -- node[above] {$z$} (C);
\draw[->] (C) -- (x);
\end{tikzpicture}
\caption{Block–diagram representation of exponential replicator dynamics (Ex–RD).}
\label{fig:exrd_block}
\end{figure}
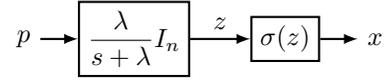
%%%%%%%%%%%%%%%% Regret %%%%%%%%%%%%%%%%%%%%%%%%
\subsection{Regret}
In continuous time, the regret of a learning dynamics with respect to a fixed
strategy $\bar{x}\in\Delta_n$ over a horizon $[0,T]$ is defined as
\[
R_T(\bar{x})
\;\coloneqq\;
\int_0^T p(t)^\top\bigl(\bar{x}-x(t)\bigr)\,dt,
\]
i.e., the difference between the cumulative reward obtained by always selecting 
$\bar{x}$ and that achieved by the trajectory $x(\cdot)$ generated by the
learning dynamics.

A learning dynamics is said to be \emph{no--regret} if
\[
\sup_{\bar{x}\in\Delta_n} R_T(\bar{x}) = o(T),
\]
equivalently, if the time--average regret $R_T(\bar{x})/T$ converges to zero as
$T\to\infty$ for every $\bar{x}\in\Delta_n$.

The dynamics is said to have \emph{finite regret} if there exists
a constant $\alpha>0$ such that, for every payoff trajectory $p(\cdot)$ and all
$T>0$,
\[
\sup_{\bar{x}\in\Delta_n} R_T(\bar{x}) \;\le\; \alpha .
\]
%%%%%%%%%%%%%% Passivity %%%%%%%%%%%%%%%%%%%%%%%%%%%%%
\subsection{Passivity}
Passivity theory provides a powerful framework for analyzing stability and
performance of interconnected dynamical systems \cite{lozano2013dissipative}. In particular, negative feedback interconnections of passive systems are guaranteed to be stable, a property that will play a central role in our analysis.

\subsubsection{Input--output operators}

A dynamical system can be viewed as an input--output operator
$H:\mathbb{U}\to\mathbb{Y}$, where $\mathbb{U},\mathbb{Y}\subset\mathbb{L}_e$
denote admissible input and output spaces, respectively.
The operator $H$ is said to be \emph{passive} if there exists a constant
$\alpha\in\mathbb{R}$ such that
\begin{equation}
\label{eq:io_passivity}
\langle Hu,u\rangle_T \;\ge\; \alpha,
\qquad
\forall\,u\in\mathbb{U},\; \forall\,T\in\mathbb{R}_+ .
\end{equation}
If the above inequality holds with $\alpha>0$, the operator is said to be
\emph{strictly passive}.

\subsubsection{State--space representation}

Consider a nonlinear system in state--space form
\begin{equation}
\label{eq:state_space_model}
\begin{aligned}
\dot{x} &= f(x,u), \qquad x(0)=x_0,\\
y &= h(x,u),
\end{aligned}
\end{equation}
where $x\in\mathbb{R}^n$ is the state, $u(t)\in\mathbb{R}^m$ is the input, and
$y(t)\in\mathbb{R}^m$ is the output.  
The system \eqref{eq:state_space_model} is passive if there exists a continuously
differentiable storage function $V:\mathbb{R}^n\to\mathbb{R}_+$ such that, for all
initial conditions $x(0)$ and all $T\in\mathbb{R}_+$,
\begin{equation}
\label{eq:ss_passivity}
V(x(T))-V(x(0))
\;\le\;
\int_0^T u(t)^\top y(t)\,dt .
\end{equation}
Equivalently, along trajectories of \eqref{eq:state_space_model},
\begin{equation}
\dot V(x)
=\nabla V(x)^\top f(x,u)
\;\le\;
u^\top y .
\end{equation}

\begin{remark}
Consider the linear time--invariant (LTI) system
\(
\dot{x}=Ax+Bu,\quad y=Cx+Du,
\)
with transfer function
$H(s)=C(sI-A)^{-1}B+D$.
The system is passive if and only if
\(
H(j\omega)+H(j\omega)^* \succeq 0
\quad \forall\,\omega\in\mathbb{R},
\)
where $(\cdot)^*$ denotes the conjugate transpose.
It is \emph{strictly passive} if there exists $\delta>0$ such that
\(
H(j\omega)+H(j\omega)^* \succeq \delta I
\quad \forall\,\omega\in\mathbb{R}.
\)

In the single--input single--output (SISO) case, passivity is equivalent to
$\operatorname{Re}\{H(j\omega)\}\ge 0$ for all $\omega\in\mathbb{R}$; such systems
are said to be positive real.
\end{remark}
%%%%%%%%%%%%%%%%%%%%%%%%%%%%%%%%%%%%%%%%%%%%%%%%%%%%%%%%%%%%%%
\subsection{Properties of the Softmax Mapping}
We recall several fundamental properties of the softmax mapping $\sigma(\cdot)$ that are used throughout the paper; see \cite{gao2017properties} and the references therein.

\begin{itemize}
    \item The softmax mapping is the gradient of the convex  log-sum-exp (lse) function:
    \[
        \sigma(v) = \nabla \operatorname{lse}(v),
    \]
    as established in Proposition~1 of~\cite{gao2017properties}.
    \item  The Jacobian of $\sigma$ is given by
    \[
        \nabla\sigma(v)
        = \mathrm{diag}(\sigma(v)) - \sigma(v)\sigma(v)^\top,
    \]
    which is symmetric and positive semidefinite. In particular, for all
    $w\in\mathbb{R}^n$,
    \(
        w^\top \nabla \sigma(v)\, w \;\ge\; 0,
    \)
    with equality if and only if
    $w \in \operatorname{span}\{\mathbf{1}_n\}$.
    Consequently, $\nabla\sigma(v)$ has rank $n-1$ and
    \(
        \ker(\nabla\sigma(v)) = \operatorname{span}\{\mathbf{1}_n\}.
    \)

    \item The softmax mapping is invariant under additive shifts along the
    all-ones direction:
    \[
        \sigma(v + c\,\mathbf{1}_n) = \sigma(v),
        \qquad \forall\, c\in\mathbb{R}.
    \]
    Thus, softmax depends only on relative score differences.
\end{itemize}
Beyond the standard properties summarized above, we establish the following
 results, which plays a key role in several of the performance
comparisons developed later.

\begin{lemma}\label{lemma:softmax_sinh}
For any $v\in\mathbb{R}^n$,
\[
v^\top\bigl(\sigma(v)-\sigma(-v)\bigr)\ \ge\ 0,
\]
with equality if and only if $v=c\,\mathbf{1}_n$ for some $c\in\mathbb{R}$.
\end{lemma}

\begin{proof}
Let
\[
    V \;\coloneqq\; v^\top\bigl(\sigma(v)-\sigma(-v)\bigr).
\]
Write $Z_1\coloneqq\sum_{j=1}^n e^{v_j}$ and $Z_2\coloneqq\sum_{j=1}^n e^{-v_j}$. Then
\begin{align*}
V
&= \sum_{i=1}^n v_i\!\left(\frac{e^{v_i}}{Z_1}-\frac{e^{-v_i}}{Z_2}\right)\\
&= \frac{1}{Z_1Z_2}\sum_{i=1}^n \sum_{j=1}^n v_i\!\left(e^{v_i-v_j}-e^{-(v_i-v_j)}\right) \\
&= \frac{1}{Z_1Z_2}\sum_{1\le j<i\le n} (v_i-v_j)\!\left(e^{v_i-v_j}-e^{-(v_i-v_j)}\right) \\
&= \frac{2}{Z_1Z_2}\sum_{1\le j<i\le n} (v_i-v_j)\,\sinh(v_i-v_j)\ \ge\ 0,
\end{align*}
because $a\,\sinh(a)\ge 0$ for all $a\in\mathbb{R}$, with equality only at $a=0$.
Hence the sum is zero if and only if $v_i - v_j = 0$ for all $i,j$, i.e., 
$v = c\,\mathbf{1}_n$ for some $c\in\mathbb{R}$.
\end{proof}
The next lemma establishes a global Lipschitz property of the softmax mapping, which is used in the local passivity arguments used later.

\begin{lemma}\label{lem:softmax_lipschitz}
The softmax mapping $\sigma$ is $\tfrac{1}{2}$-Lipschitz with respect to the
Euclidean norm; that is, for all $u,v \in \mathbb{R}^n$,
\begin{equation}
    \label{eq:softmax_lipschitz}
    \|\sigma(u+v) - \sigma(u)\|_2
    \;\le\;
    \frac{1}{2}\,\|v\|_2.
\end{equation}
\end{lemma}

\begin{proof}
By the mean-value theorem for vector-valued functions,
\[
\sigma(u+v) - \sigma(u)
= \int_0^1 \nabla\sigma(u + s v)\, v \, ds.
\]
Taking the Euclidean norm on both sides and using the submultiplicativity property of the induced matrix norm gives
\begin{equation*}
    \begin{aligned}
        \|\sigma(u+v) - \sigma(u)\|_2
&\le \int_0^1 \|\nabla\sigma(u + s v)\, v\|_2\, ds \\
&\le \int_0^1 \|\nabla\sigma(u + s v)\|_2\, ds\, \|v\|_2.
    \end{aligned}
\end{equation*}
For any $s = \sigma(z) \in \Delta_n$, one has
\[
\|\nabla\sigma(z)\|_2
= \|\mathrm{diag}(s) - s s^\top\|_2
\le \frac{1}{2}.
\]
This follows since each row of $\nabla\sigma(z)$ satisfies
\[
\sum_{j=1}^n |[\nabla\sigma(z)]_{ij}| = 2 s_i (1-s_i) \le \frac{1}{2},
\]
and by symmetry, $\|J_\sigma(z)\|_1 = \|J_\sigma(z)\|_\infty \le \tfrac{1}{2}$,
implying $\|\nabla\sigma(z)\|_2 \le \sqrt{\|\nabla\sigma(z)\|_1 \|\nabla\sigma(z)\|_\infty} \le \tfrac{1}{2}$.
Hence,
\[
\|\sigma(u+v) - \sigma(u)\|_2
\le \int_0^1 \frac{1}{2}\, ds\, \|v\|_2
= \frac{1}{2}\, \|v\|_2.
\]
\end{proof}
%%%%%%%%%%%%%%%%%%%%%%%%%%%%%%%%%%%%%%%%%%%%%%%%
\section{Motivation}\label{sec:motivation}

In \cite{abdelraouf2025passivity}, a fundamental connection is established between passivity and regret guarantees for continuous–time learning dynamics. In particular, it is shown that if a learning dynamic model is passive from the payoff input $p$ to the output $x-\bar{x}$ for every fixed strategy $\bar{x}\in\Delta_n$, then the learning dynamic has finite regret.
Equivalently, there exists a constant $C>0$ such that, for every payoff trajectory $p(\cdot)$ and every horizon $T>0$,
\[
\sup_{\bar{x}\in\Delta_n}\int_0^T p(t)^\top(\bar{x}-x(t))\,dt \;\le\; C,
\]
Equivalently, the regret with respect to any fixed strategy
$\bar{x}\in\Delta_n$ is uniformly bounded over time.

Using this framework, RD \eqref{eq:RD} is shown to satisfy the required passivity condition and hence to guarantee finite regret. In contrast, BNN dynamics \eqref{eq:BNN}, Smith dynamics \eqref{eq:Smith}, TP dynamics \eqref{eq:TP}, and Ex–RD violate this passivity condition. Specifically, for each of these dynamics, one can construct a payoff trajectory $p(\cdot)$ and select a fixed strategy $\bar{x}\in\Delta_n$ such that passivity from $p$ to $x-\bar{x}$ fails. Consequently, these learning
rules do not have finite regret.

Finite regret therefore provides a strong performance guarantee: asymptotically, the learner achieves an average reward no worse than that obtained by any fixed strategy in the simplex.

A natural question then arises: \emph{Is finite regret, by itself, a reliable proxy for performance as measured by asymptotic average reward?} The following two examples demonstrate that finite regret guarantees alone are \emph{insufficient} to assess long--run performance: there exist environments in which a learning rule that does not have finite regret achieves a strictly higher
asymptotic average reward than a finite--regret algorithm, and conversely, environments in which a finite--regret algorithm outperforms learning rules that do not have finite regret.

\begin{example}[Finite–regret dominance {\cite{abdelraouf2025passivity}}]
\label{ex:RD is better}
Consider the payoff trajectory
\[
p(t)=\begin{bmatrix}\sin t \\ 0.5\end{bmatrix}.
\]
For this environment, BNN, Smith, and TP dynamics achieve asymptotic average rewards of approximately $0.374$, $0.453$, and $0.467$, respectively, each strictly lower than the reward achieved by RD, which has finite regret. See Fig.~\ref{fig:avg_reward_RD_better}.
\end{example}

\begin{figure}[H]
    \centering
    \includegraphics[width=0.6\linewidth]{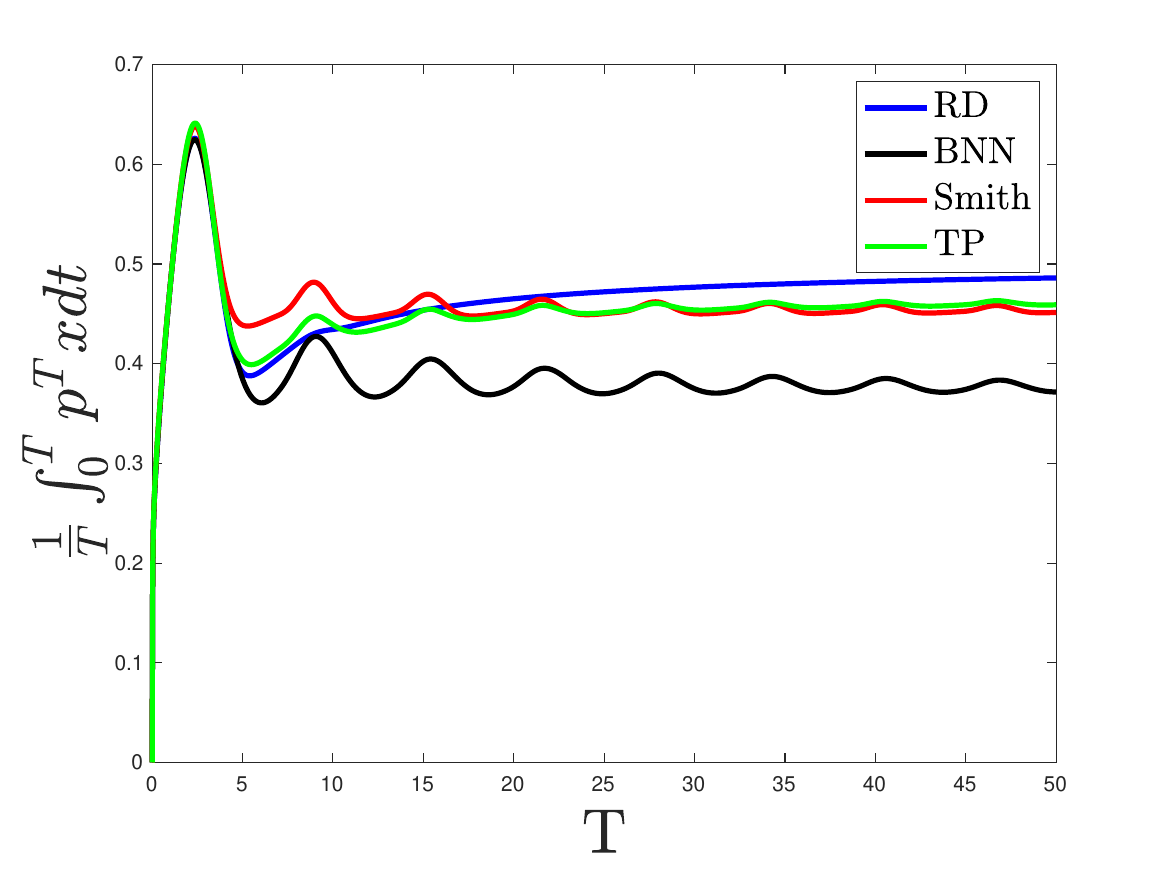}
    \caption{Performance of different learning dynamics in the environment
  \(p(t)=\begin{bmatrix}\sin t & 0.5\end{bmatrix}^\top\).}
  \label{fig:avg_reward_RD_better}
\end{figure}

\begin{example}
\label{ex:RD is worse}
Now consider the payoff trajectory
\[
p(t)=\begin{bmatrix}\sin t \\ -\sin t\end{bmatrix}.
\]
Under RD, the asymptotic average reward equals $0$. In
contrast, BNN, Smith, and TP dynamics achieve approximately $0.18$, $0.33$, and $0.18$, respectively—each strictly higher than the value obtained by RD. See Fig.~\ref{fig:avg_reward_RD_worse}.
\end{example}

\begin{figure}[H]
  \centering
  \includegraphics[width=0.6\linewidth]{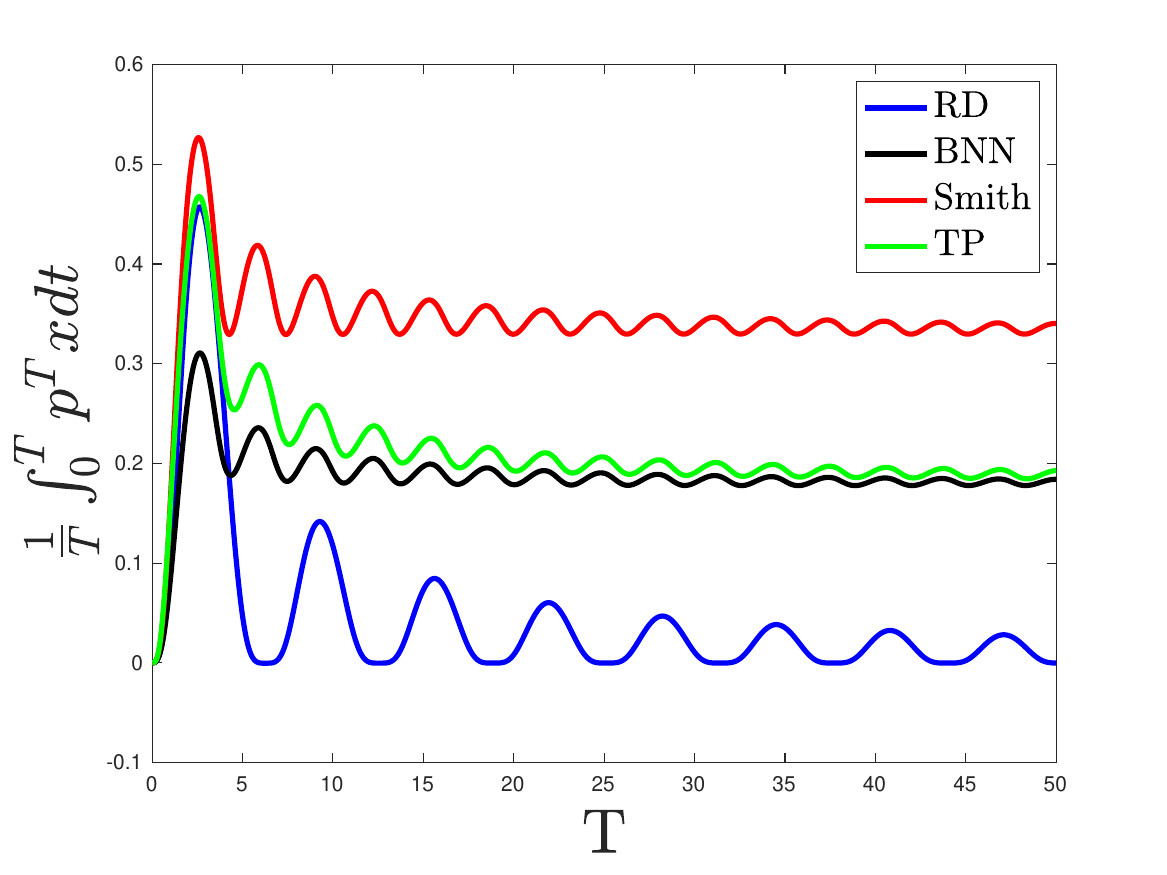}
  \caption{Performance of different learning dynamics in the environment
  \(p(t)=\begin{bmatrix}\sin t &-\sin t\end{bmatrix}^\top\).}
  \label{fig:avg_reward_RD_worse}
\end{figure}

Taken together, these examples show that while finite regret guarantees that the asymptotic average reward of a learning rule is no worse than that of any fixed strategy in the simplex, there exist payoff environments in which learning dynamics that do \emph{not} have finite regret achieve strictly higher asymptotic average rewards. Consequently, finite regret alone does not ensure strong performance across all payoff environments. 

This observation motivates the search for learning algorithms that combine
finite regret guarantees with uniformly superior performance across payoff
environments, and naturally leads to the following question.

\medskip
\noindent\emph{Does there exist a ``free lunch'' among no--regret algorithms?}
That is, can one find two no--regret learning algorithms \(A\) and \(B\) such that \(B\) achieves \emph{at least} as large an asymptotic average reward as \(A\) in every payoff environment, and a strictly larger reward in
some environment?

This paper addresses this question by introducing a notion of performance dominance between learning dynamics and by comparing the cumulative rewards achieved by different no--regret algorithms. To formalize this comparison, we introduce the following definitions.

\begin{definition}[Uniform Dominance]
Let $p(\cdot)$ be a payoff trajectory, and let $x_{\mathrm{L1}}(\cdot)$ and 
$x_{\mathrm{L2}}(\cdot)$ denote the strategies generated by learning dynamics 
$\mathrm{L1}$ and $\mathrm{L2}$, respectively, when subjected to the same payoff signal. $\mathrm{L1}$ is said to uniformly dominate $\mathrm{L2}$ if, for every horizon $T>0$,
\[
    \int_{0}^{T} p(t)^\top x_{\mathrm{L1}}(t)\, dt
    \;\ge\;
    \int_{0}^{T} p(t)^\top x_{\mathrm{L2}}(t)\, dt .
\]
That is, $\mathrm{L1}$ achieves a weakly higher cumulative reward than $\mathrm{L2}$ at every finite time.
\end{definition}

\begin{definition}[Asymptotic Dominance]
Let $p(\cdot)$ be a payoff trajectory, and let $x_{\mathrm{L1}}(\cdot)$ 
and $x_{\mathrm{L2}}(\cdot)$ denote the corresponding strategies produced 
by learning dynamics $\mathrm{L1}$ and $\mathrm{L2}$ under the same payoff signal.
$\mathrm{L1}$ is said to asymptotically dominate $\mathrm{L2}$ if 
\[
    \liminf_{T\to\infty}
    \frac{1}{T}\int_{0}^{T} p(t)^\top x_{\mathrm{L1}}(t)\, dt
    \;\ge\;
    \liminf_{T\to\infty}
    \frac{1}{T}\int_{0}^{T} p(t)^\top x_{\mathrm{L2}}(t)\, dt .
\]
In this case, $\mathrm{L1}$ achieves a weakly higher long-run average reward than $\mathrm{L2}$.
\end{definition}

Motivated by this perspective, we search for learning dynamics that preserve
finite regret guarantees while achieving systematically improved performance
in the sense of uniform or asymptotic dominance. In this paper, we pursue this
goal by studying payoff--based higher--order extensions of replicator dynamics.

%%%%%%%%%%%%%%%%%%%%%%%%%%%%%%%%%%%%%%%%%%%%%%%%%%%%%%%%%%%%%%%%%
\section{Payoff-Based Higher-Order replicator dynamics}\label{sec:pay_off_based_higher_order_RD}

In this section, we introduce a payoff--based, higher--order variant of the
replicator dynamics (RD). The design follows a classical idea from feedback
control, namely anticipatory (lead) compensation
\cite{shamma2005dynamic,arslan2004distributed,arslan2006anticipatory}.
Such higher--order modifications have been shown to enable convergence to
mixed--strategy Nash equilibria in settings where first--order learning
dynamics fail to converge \cite{toonsi2023higher}. In the anticipatory RD,
rather than feeding the instantaneous payoff
$p(t)\in\mathbb{R}^n$ directly into the standard RD, we apply a simple LTI system
that produces an anticipated payoff $\hat p(t)$:
\begin{equation}
\begin{aligned}
    \dot q &= \lambda\bigl(p-q\bigr),\\
    \hat p &= p + \gamma\,\dot q \;=\; -\,\gamma\lambda\,q + (1+\gamma\lambda)\,p,
\end{aligned}
\label{eq:anticipatroy_model}
\end{equation}
with state $q\in\mathbb{R}^n$ and parameters $\lambda,\gamma>0$.
The corresponding (diagonal) transfer function matrix from $p$ to $\hat p$ is
\[
\frac{(1+\gamma\lambda)s+\lambda}{s+\lambda} I_n.
\]

Hence, the anticipatory RD is characterized as a cascade interconnection between the LTI system \eqref{eq:anticipatroy_model} and the standard RD, yielding
\begin{equation}
\begin{aligned}
    \dot z &= p + \gamma\lambda\,(p-q)=\hat{p},\\
    \dot q &= \lambda\,(p-q),\\
    x &= \sigma(z),
\end{aligned}
\label{eq:anticipatory_RD}
\end{equation}
where $z(t)$ denotes the score signal.
A block-diagram representation of the anticipatory RD
\eqref{eq:anticipatory_RD} is shown in
Figure~\ref{fig:anticipatory_RD_block}.

\begin{figure}[H]
    \centering
    \includegraphics[width=0.7\linewidth]{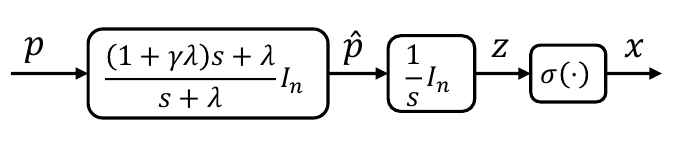}
    \caption{Block diagram representation for the anticipatory RD.}
    \label{fig:anticipatory_RD_block}
\end{figure}

To generalize the higher--order variant of RD, we introduce the predictive (optimistic) RD, in which the score signal is given by
\(
z(t) = \int_0^t p(\tau)\, d\tau + m(t),
\)
where $m(t)$ is a prediction of the payoff $p(t)$. The transfer function matrix mapping $p$ to $m$ is $H(s)=h(s)I_n$, corresponding to identical prediction models for each payoff channel, where
\(
h(s)=C_h(sI_m-A_h)^{-1}B_h,
\)
with $A_h\in\mathbb{R}^{m\times m}$,
$B_h\in\mathbb{R}^{m\times 1}$, and
$C_h\in\mathbb{R}^{1\times m}$.
Depending on the choice of $H(s)$, the signal $m(t)$ can be interpreted
as a predictor of the instantaneous payoff. In this sense, and for
terminological simplicity, we will refer to $H(s)$ as a predictor
throughout the paper, even though the analysis applies to general
stable linear system mapping  $p$ to $m$.
We now write a state-space realization of the resulting predictive
RD.  The dynamic model of the predictive RD is
\begin{equation}
    \begin{aligned}
        \dot{r} & = p,\\
        \dot{x}_h &= (A_h \otimes I_n)\,x_h + (B_h \otimes I_n)\,p,\\
        m &= (C_h \otimes I_n)\,x_h,\\
        z &= r + m,\\
        x &= \sigma(z),
    \end{aligned}
    \label{eq:predictive_RD}
\end{equation}
where $r\in\mathbb{R}^n$ is the score of the standard RD,
$x_h\in\mathbb{R}^{mn}$ is the predictor hidden state vector, and
$m\in\mathbb{R}^n$ is the predicted payoff signal.
A block-diagram representation of the predictive RD, with $h(s)I_n$ as the
prediction transfer function matrix, is shown in
Figure~\ref{fig:predictive_RD_block}.

\begin{figure}[H]
    \centering
    \includegraphics[width=0.6\linewidth]{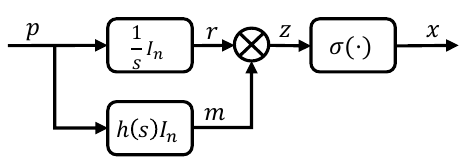}
    \caption{Block diagram representation for predictive RD.}
    \label{fig:predictive_RD_block}
\end{figure}

\begin{remark}\label{rem:anticipatory_predicitive}
In \eqref{eq:predictive_RD}, setting
$A_h=-\lambda$, $B_h=\gamma\lambda$, and $C_h=1$ yields
$\dot m=\gamma\lambda p-\lambda m$, corresponding to a first--order low--pass
predictor with transfer function $h(s)=\tfrac{\gamma\lambda}{s+\lambda}$.
In this case, the transfer function matrix mapping the payoff $p$ to the score
$z$ is
\[
\left(\tfrac{1}{s}+\tfrac{\gamma\lambda}{s+\lambda}\right)I_n
=\tfrac{(1+\gamma\lambda)s+\lambda}{s(s+\lambda)}I_n,
\]
which coincides with the transfer function mapping $p$ to $z$ in the
anticipatory RD \eqref{eq:anticipatory_RD}. The resulting dynamics are
\begin{equation}
    \begin{aligned}
        \dot{r} & = p,\\
        \dot{m} & = \gamma\lambda p-\lambda m,\\
        z & = r + m,\\
        x & = \sigma(z).
    \end{aligned}
    \label{eq:predictive_RD_low_pass}
\end{equation}
From \eqref{eq:predictive_RD_low_pass}, we obtain
$\dot z=p+\gamma\lambda p-\lambda m$. Setting $m=\gamma q$ recovers the
anticipatory RD \eqref{eq:anticipatory_RD}. Therefore, predictive RD with a
low--pass predictor is dynamically equivalent to the anticipatory RD
\eqref{eq:anticipatory_RD}, provided that $m(0)=\gamma q(0)$.
\end{remark}
Motivated by the formulation of the predictive RD, we study an idealized
setting with perfect payoff information.
%%%%%%%%%%%%%%%%%%%%%%%%%%%%%%%%%%%%%%%%%%%%%%%%%%%%%%%%%%%%%%%%%%%%%%
\section{Oracle Replicator Dynamics (Oracle RD)}\label{sec:oracle_RD}

We consider a special case of the predictive RD \eqref{eq:predictive_RD} in
which the learner has perfect access to the instantaneous payoff vector, i.e., $m(t)=p(t)$. We refer to this model as the oracle RD. For the case $H(s)=I_n$, the associated score variable equals the cumulative payoff augmented by its instantaneous value:
\[
z(t) \;=\; \int_{0}^{t} p(\tau)\,d\tau \;+\; p(t).
\]
The resulting oracle RD model can be written as
\begin{equation}
\label{eq:oracle_RD}
\begin{aligned}
\dot r &= p, \\
z &= r+ p, \\
x &= \sigma\bigl(z\bigr).
\end{aligned}
\end{equation}
We next show that, for any payoff trajectory $p(\cdot)$, the oracle RD uniformly dominates the standard RD.

\begin{theorem}
\label{thm:oracle_dominates_RD}
Consider the standard RD \eqref{eq:RD} and the oracle
RD \eqref{eq:oracle_RD} driven by the same payoff trajectory
$p(\cdot)$. Let $z_{\mathrm{R}}(\cdot)$ and $x_{\mathrm{R}}(\cdot)$ denote the score and strategy of the standard RD, and let
$z_{\mathrm{OR}}(\cdot)$ and $x_{\mathrm{OR}}(\cdot)$ denote the corresponding quantities for the oracle RD. If the initial conditions are matched as \( r(0) = z_{\mathrm{R}}(0) \), then the oracle RD uniformly dominates the standard RD, i.e.,
\[
\int_{0}^{T} p(t)^\top x_{\mathrm{OR}}(t)\,dt
\;\ge\;
\int_{0}^{T} p(t)^\top x_{\mathrm{R}}(t)\,dt
\quad \text{for all } T \ge 0.
\]
\end{theorem}

\begin{proof}
Under the common payoff trajectory $p(\cdot)$, the score and strategy of the standard RD satisfy
\[
z_{\mathrm{R}}(t)
= z_{\mathrm{R}}(0) + \int_{0}^{t} p(\tau)\,d\tau,
\qquad
x_{\mathrm{R}}(t) = \sigma\bigl(z_{\mathrm{R}}(t)\bigr).
\]
Similarly, the oracle RD satisfies
\[
z_{\mathrm{OR}}(t)
= r(0) + \int_{0}^{t} p(\tau)\,d\tau + p(t),
\quad
x_{\mathrm{OR}}(t) = \sigma\bigl(z_{\mathrm{OR}}(t)\bigr).
\]
With the initialization $r(0)=z_{\mathrm{R}}(0)$, it follows that
\[
z_{\mathrm{OR}}(t) - z_{\mathrm{R}}(t) = p(t)
\qquad \text{for all } t \ge 0.
\]
Define the cumulative--reward gap
\begin{equation*}
    \begin{aligned}
        \Delta P(T) &=
\int_{0}^{T}
p(t)^\top\!\bigl(x_{\mathrm{OR}}(t) - x_{\mathrm{R}}(t)\bigr)\,dt \\
        &=
\int_{0}^{T}
p(t)^\top
\bigl(
\sigma(z_{\mathrm{OR}}(t)) - \sigma(z_{\mathrm{R}}(t))
\bigr)\,dt .
    \end{aligned}
\end{equation*}
By the mean--value theorem for vector--valued functions,
\[
\sigma(z_{\mathrm{OR}})-\sigma(z_{\mathrm{R}})
= \int_{0}^{1}
\nabla\sigma \bigl(z_{\mathrm{R}}+s(z_{\mathrm{OR}}-z_{\mathrm{R}})\bigr)\,
(z_{\mathrm{OR}}-z_{\mathrm{R}})\,ds .
\]
Using $z_{\mathrm{OR}}(t)-z_{\mathrm{R}}(t)=p(t)$, we obtain
\[
\Delta P(T)
=
\int_{0}^{T}
p(t)^\top
\left(
\int_{0}^{1}
\nabla\sigma\!\bigl(z_{\mathrm{R}}(t)+s\,p(t)\bigr)\,ds
\right)
p(t)\,dt .
\]

Since $\nabla\sigma(\cdot)$ is positive semidefinite and convex combinations of positive semidefinite matrices remain positive semidefinite, the integrand is nonnegative for every $t$. Hence,
\[
\Delta P(T) \ge 0 \qquad \text{for all } T \ge 0,
\]
which establishes the uniform dominance claim.

Moreover, equality holds only if $p(t)$ lies in the nullspace of
$\nabla\sigma(\cdot)$ along the entire trajectory, i.e., if
$p(t) \in \operatorname{span}\{\mathbf{1}_n\}$ almost everywhere. This completes
the proof.
\end{proof}
Although Theorem~\ref{thm:oracle_dominates_RD} may be viewed as a first
``free lunch'' result, exhibiting a learning rule that uniformly outperforms
another across all payoff environments, the oracle RD relies on perfect,
instantaneous access to the payoff vector and is therefore not implementable in
practice. Indeed, if the payoff were fully known at each time $t$, the learner
could simply select the maximizer with probability one and achieve the highest
possible instantaneous reward.

Nevertheless, the oracle RD plays an important conceptual role. It admits a
simple representation that makes explicit the role of prediction in
improving performance. This observation motivates the frequency--domain
viewpoint developed next, which links the performance of the learning dynamics directly to the gain and phase characteristics of the transfer function mapping payoffs to scores.

%%%%%%%%%%%%%%%%%%%%%%%%%%%%%%%%%%%%%%%%%%%%%%%%%%%%%%%%%%%%%%%%%%%%%%%%%
\section{Frequency--Domain Intuition (Bode View)}\label{sec:frequency_domain_intuition}

The dominance result for oracle RD admits a natural interpretation from a
frequency--domain perspective. In this section, we provide a Bode--plot view
that offers qualitative intuition for why anticipatory or predictive structure
leads to improved performance.

For the standard RD and the oracle RD, the linear mappings from the payoff
signal $p$ to the score signal $z$ can be written as diagonal transfer function
matrices
\[
G_{\mathrm{R}}(s) = g_{\mathrm{R}}(s)\, I_n, \qquad
G_{\mathrm{OR}}(s) = g_{\mathrm{OR}}(s)\, I_n,
\]
where
\[
g_{\mathrm{R}}(s) = \frac{1}{s}, \qquad
g_{\mathrm{OR}}(s) = 1 + \frac{1}{s}.
\]

The corresponding frequency responses are given by
\[
g_{\mathrm{R}}(j\omega) = -\frac{j}{\omega}, \qquad
g_{\mathrm{OR}}(j\omega) = 1 - \frac{j}{\omega}, \qquad \omega > 0,
\]
with magnitudes
\[
|g_{\mathrm{R}}(j\omega)| = \frac{1}{\omega}, \qquad
|g_{\mathrm{OR}}(j\omega)| = \sqrt{1 + \frac{1}{\omega^2}},
\]
and phases
\[
\arg(g_{\mathrm{R}}(j\omega)) = -\frac{\pi}{2}, \qquad
\arg(g_{\mathrm{OR}}(j\omega)) = -\arctan\!\Bigl(\frac{1}{\omega}\Bigr).
\]

These expressions show that, at every nonzero frequency, oracle RD has
both a strictly larger gain and a strictly smaller phase lag than the standard
RD. In particular, the phase lag of oracle RD is always less than $\pi/2$.

These frequency--domain properties provide direct intuition for the observed
performance improvement. Performance improves when the score signal is both
better aligned with the payoff and of larger magnitude. For sinusoidal payoff
inputs, the reduced phase lag implies that the oracle RD score signal tracks the
payoff more closely in time.

Because the softmax mapping $\sigma(\cdot)$ is static, this improved alignment
of the score signal carries over directly to the strategy signal
$x = \sigma(z)$. As a result, the strategy produces a larger time--averaged
projection onto the payoff direction, yielding a higher accumulated reward.

In addition, the increased gain of oracle RD amplifies variations in the score
signal relative to RD. By monotonicity of the softmax mapping, larger score
magnitudes concentrate probability mass more heavily on the instantaneous
maximizers of the payoff. This concentration effect increases the instantaneous
reward $p(t)^\top x(t)$ and further contributes to improved performance.

Motivated by this intuition, we next make the relationship between frequency
response of the transfer function associated with each learning dynamics and its performance precise. We consider learning dynamics characterized by
a cascade interconnection between a diagonal transfer function matrix
$G(s)=g(s)I_n$ and the softmax mapping $\sigma(\cdot)$, where $g(s)$ is a passive transfer function. We compare their asymptotic average reward against that of the standard RD under sinusoidal payoff environments of the form
\[
p(t) = \bar p\,\sin(\omega t) + \bar q\,\cos(\omega t),
\]
with $\bar p,\bar q \in \mathbb{R}^n$ and $\omega>0$. The following theorems show explicitly how the gain and phase lag of $g(j\omega)$ determine the asymptotic average reward achieved by the learning dynamics.
%%%%%%%%%%%%%%%%%%%%%%%%%%%%%%%%%%%%%%%%%%%%%%%%%%%%%%%%%%%%%%%
\subsection{Dominance over Replicator Dynamics}

\begin{theorem}
\label{thm:passive_vs_RD_sinusoid}
Fix $\omega>0$ and $\bar p,\bar q\in\mathbb{R}^n$. Consider the class of learning dynamics defined by the cascade interconnection of a diagonal LTI system $G(s)=g(s)I_n$ with the softmax mapping $\sigma(\cdot)$, where $g(s)$ is a (passive) transfer function. Under the sinusoidal payoff environment
\[
p(t)=\bar p\,\sin(\omega t)+\bar q\,\cos(\omega t),
\]
let $J_g(\omega)$ denotes the asymptotic average reward achieved by this learning dynamics, and let $J_{\mathrm{RD}}(\omega)$ denote the corresponding asymptotic average reward of the standard replicator dynamics (RD). Then
\[
J_g(\omega) \;\ge\; J_{\mathrm{RD}}(\omega).
\]
Moreover, the inequality is strict whenever $g(s)$ is strictly passive and at least one of $\bar p,\bar q$ does not belong to $\operatorname{span}\{\mathbf{1}_n\}$.
\end{theorem}

\begin{proof}
    Apply the payoff $p(t)=\bar p\sin(\omega t)+\bar q\cos(\omega t)$ to the cascade $G(s)=g(s)I_n$ (from $p$ to $z$) followed by the static softmax mapping $x=\sigma(z)$. In steady state,
    \[
z(t)
=
a\,\bar p\,\sin(\omega t-\phi)
+
a\,\bar q\,\cos(\omega t-\phi),
\]
where
\[
a \coloneqq |g(j\omega)|>0,
\qquad
\phi \coloneqq -\arg g(j\omega).
\]
Define the asymptotic average reward
\[
J_g(\omega)
\;\coloneqq\;
\frac{\omega}{2\pi}
\int_{0}^{2\pi/\omega}
\bigl(\bar p\sin(\omega t)+\bar q\cos(\omega t)\bigr)^\top
\sigma\!\bigl(z(t)\bigr)\,dt .
\]
Introducing the change of variables $u=\omega t$ yields
\begin{equation*}
\begin{aligned}
J_g(\phi,a)
&=
\frac{1}{2\pi}
\int_{0}^{2\pi}
(\bar p\sin u+\bar q\cos u)^\top \\
&\qquad\qquad 
\sigma \bigl(
a\,\bar p\sin(u-\phi)+a\,\bar q\cos(u-\phi)
\bigr)\,du .
\end{aligned}
\end{equation*}
Next, let $\tau=u-\phi$. Using $2\pi$-periodicity and the identities
\begin{equation*}
    \begin{aligned}
        \sin(u)&=\sin\tau\cos\phi+\cos\tau\sin\phi,\\
        \cos(u)&=\cos\tau\cos\phi-\sin\tau\sin\phi,
    \end{aligned}
\end{equation*}
define
\[
v(\tau)\coloneqq \bar p\,\sin\tau+\bar q\,\cos\tau,
\qquad
w(\tau)\coloneqq \bar p\,\cos\tau-\bar q\,\sin\tau .
\]
Then
\begin{equation*}
\begin{aligned}
J_g(\phi,a)
&=
\frac{1}{2\pi}
\int_{0}^{2\pi}
\bigl(v(\tau)\cos\phi+w(\tau)\sin\phi\bigr)^\top
\sigma\!\bigl(a\,v(\tau)\bigr)\,d\tau .
\end{aligned}
\end{equation*}
Consequently,
\[
J_g(\phi,a)
=
\cos\phi\,T_1(a)
+
\sin\phi\,T_2(a),
\]
where
\begin{equation*}
    \begin{aligned}
        T_1(a) & = 
\frac{1}{2\pi}
\int_{0}^{2\pi}
v(\tau)^\top \sigma \bigl(a\,v(\tau)\bigr)\,d\tau, \\
T_2(a)&=
\frac{1}{2\pi}
\int_{0}^{2\pi}
w(\tau)^\top \sigma \bigl(a\,v(\tau)\bigr)\,d\tau .
    \end{aligned}
\end{equation*}
\textit{Analysis of $T_2(a)$.}
Since $v'(\tau)=w(\tau)$ and $\sigma(\cdot)=\nabla\mathrm{lse}(\cdot)$, the chain
rule yields
\[
\frac{d}{d\tau}\,\mathrm{lse} \bigl(a\,v(\tau)\bigr)
=
a\,w(\tau)^\top \sigma \bigl(a\,v(\tau)\bigr).
\]
Therefore,
\begin{equation*}
\begin{aligned}
T_2(a)
&=
\frac{1}{2\pi}\int_{0}^{2\pi}
w(\tau)^\top \sigma\!\bigl(a\,v(\tau)\bigr)\,d\tau \\
&=
\frac{1}{2\pi a}
\Bigl[\mathrm{lse}\!\bigl(a\,v(\tau)\bigr)\Bigr]_{0}^{2\pi}
=0,
\end{aligned}
\end{equation*}
since $v(\tau)$ is $2\pi$-periodic.

\textit{Analysis of $T_1(a)$.} Let $\tilde{v}(\tau) = \bar p \sin \tau -\bar{q} \cos \tau$, then 
\small{
\begin{equation*}
    \begin{aligned}
        T_1(a) &= \frac{1}{2\pi} \int_0^{2\pi} v(\tau)^\top \sigma(av(\tau)) d\tau \\
        &=  \int_0^{\pi/2} v(\tau)^\top \sigma(av(\tau)) d\tau  + \int_0^{\pi/2} \tilde{v}(\tau)^\top \sigma(a\tilde{v}(\tau)) d\tau \\
        &+ \int_0^{\pi/2} -v(\tau)^\top \sigma(-av(\tau)) d\tau  + \int_0^{\pi/2} -\tilde{v}(\tau)^\top \sigma(-a\tilde{v}(\tau)) d\tau
    \end{aligned}
\end{equation*}}

Then $T_1(a)$ can be written as 
\begin{equation*}
    \begin{aligned}
        T_1(a) & = \frac{1}{2\pi a} \int_0^{\pi/2} av(\tau)^\top (\sigma(av(\tau))-\sigma(-av(\tau))) d \tau \\
        &+\frac{1}{2\pi a} \int_0^{\pi/2} a\tilde{v}(\tau)^\top (\sigma(a\tilde{v}(\tau))-\sigma(-a\tilde{v}(\tau))) d \tau 
    \end{aligned}
\end{equation*}
By Lemma~\ref{lemma:softmax_sinh}, each integrand is nonnegative and strictly
positive whenever $v(\tau),\tilde{v}(\tau)\notin\operatorname{span}\{\mathbf{1}_n\}$. If at least one of $\bar p$ or $\bar q$ does not belong to
$\operatorname{span}\{\mathbf{1}_n\}$, continuity implies the existence of an
interval of $\tau$ with nonzero measure on which this condition holds, and hence
$T_1(a)>0$. Otherwise, if $\bar p,\bar q\in\operatorname{span}\{\mathbf{1}_n\}$,
then $p(t)\in\operatorname{span}\{\mathbf{1}_n\}$ for all $t$ and $T_1(a)=0$. Hence, 

\[
J_g(\phi,a)=T_1(a)\cos\phi,
\qquad
T_1(a)\ge 0.
\]
Passivity of $g$ implies $\text{Re}\{g(j\omega)\}\ge 0$, and hence
\[
\cos\phi
=
\frac{\text{Re}\{g(j\omega)\}}{|g(j\omega)|}
\ge 0,
\]
which yields $J_g(\omega)\ge 0$. For standard RD,
$g_{\mathrm{RD}}(j\omega)=-j/\omega$, so $\phi_{\mathrm{RD}}=\pi/2$ and
$J_{\mathrm{RD}}=0$. Therefore,
\[
J_g(\omega)\ge J_{\mathrm{RD}}(\omega),
\]
with strict inequality whenever $T_1(a)>0$ and $g(s)$ is strictly passive.
\end{proof}

\subsection{Phase-Lag Dominance at Fixed Gain}

The representation $J_g(\omega)=T_1(a)\cos\phi$ obtained in
Theorem~\ref{thm:passive_vs_RD_sinusoid} isolates the respective roles of the
gain $a=|g(j\omega)|$ and phase lag $\phi=-\arg(g(j\omega))$.
We first consider comparisons at fixed gain.

\begin{corollary}
\label{cor:phase_dominance_fixed_gain}
Fix $\omega>0$ and $\bar p,\bar q\in\mathbb{R}^n$. Consider two learning dynamics
obtained by cascading $g_1(s)I_n$ and $g_2(s)I_n$ with the softmax mapping
$\sigma(\cdot)$, respectively, where $g_1$ and $g_2$ are passive transfer
functions. Suppose that
\begin{equation*}
    \begin{aligned}
        |g_1(j\omega)|&=|g_2(j\omega)|=a>0,\\
        \phi_k &\coloneqq -\arg g_k(j\omega), \quad k=1,2.
    \end{aligned}
\end{equation*}
Let $J_{g_k}(\omega)$ denote the asymptotic average reward achieved by the
$k$-th learning dynamics at frequency $\omega$. Then
\[
J_{g_k}(\omega)=T_1(a)\cos\phi_k .
\]
In particular, if $\cos\phi_1>\cos\phi_2$, then $J_{g_1}(\omega)>J_{g_2}(\omega)$
whenever $T_1(a)>0$, i.e., whenever at least one of $\bar p,\bar q$ does not
belong to $\operatorname{span}\{\mathbf{1}_n\}$. If $\cos\phi_1=\cos\phi_2$, then
$J_{g_1}(\omega)=J_{g_2}(\omega)$.
\end{corollary}

By Theorem~\ref{thm:passive_vs_RD_sinusoid}, any learning dynamics of the form
$\sigma\circ g(s)I_n$ with passive $g(\cdot)$ achieves an asymptotic average
reward no smaller than that of standard RD(corresponding to
$g(s)=1/s$) in the sinusoidal environment
$p(t)=\bar p\sin(\omega t)+\bar q\cos(\omega t)$ for any $\omega>0$, provided
that at least one of $\bar p,\bar q$ does not belong to
$\operatorname{span}\{\mathbf{1}_n\}$. Moreover, the inequality is strict whenever
$g(\cdot)$ is strictly positive real at $\omega$.

Corollary~\ref{cor:phase_dominance_fixed_gain} further shows that, among two
learning dynamics of the form $\sigma\circ g_1(s)I_n$ and
$\sigma\circ g_2(s)I_n$ with identical gain
$|g_1(j\omega)|=|g_2(j\omega)|$, the one with the smaller phase lag (or,
equivalently, the larger value of $\cos\phi$) at the specified frequency
achieves the larger asymptotic average reward, whenever $T_1(a)>0$.

\subsection{Gain Monotonicity at Fixed Phase}

We next examine how the gain $a=|g(j\omega)|$ influences performance when the
phase lag $\phi$ is fixed.
\begin{proposition}
\label{prop:gain_monotonicity_fixed_phase}
Fix $\omega>0$ and let
$p(t)=\bar p\,\sin(\omega t)+\bar q\,\cos(\omega t)$.
For the learning dynamics $\sigma\circ g(s)I_n$ with phase lag
$\phi=-\arg(g(j\omega))$ and gain $a=|g(j\omega)|$, define the asymptotic
average reward $J_g(\phi,a)$. Then
\[
\frac{\partial J_g(\phi,a)}{\partial a}\ge 0.
\]
Moreover, for any fixed $\phi$ with $\cos\phi\ge 0$, $J_g(\phi,a)$ is
nondecreasing in $a$, and is strictly increasing when $\cos\phi>0$ and at least
one of $\bar p$ or $\bar q$ does not belong to
$\operatorname{span}\{\mathbf{1}_n\}$.
\end{proposition}

\begin{proof}
By Theorem~\ref{thm:passive_vs_RD_sinusoid},
\begin{equation*}
    \begin{aligned}
        J_g(\phi,a)= \frac{1}{2\pi} \int_0^{2\pi} &\bigl(\bar p\sin\tau+\bar q\cos\tau\bigr)^\top  \\
        &\sigma \bigl(a\,\bar p \sin (\tau-\phi)+a\,\bar q \cos(\tau-\phi)\bigr)\, d\tau.
    \end{aligned}
\end{equation*}
Let
\[
u(\tau)=\bar p\sin\tau+\bar q\cos\tau,\qquad
v(\tau)=\bar p \sin (\tau-\phi)+\bar q \cos(\tau-\phi).
\]
Then
\[
v'(\tau)=\bar p \cos(\tau-\phi)- \bar q \sin(\tau-\phi),
\]
and, using angle–addition identities,
\[
\begin{aligned}
u(\tau) &= \cos\phi\, v(\tau) + \sin\phi\, v'(\tau).
\end{aligned}
\]
Differentiate under the integral and apply the chain rule:
\[
\begin{aligned}
\frac{\partial J_g}{\partial a}(\phi,a)
&=\frac{1}{2\pi}\int_{0}^{2\pi} u(\tau)^\top \nabla\sigma\!\big(a\,v(\tau)\big)\,v(\tau)\,d\tau \\
&=\cos\phi\,I_1(a)+\sin\phi\,I_2(a),
\end{aligned}
\]
with
\[
I_1(a)=\frac{1}{2\pi}\int_{0}^{2\pi} v(\tau)^\top \nabla\sigma\!\big(av(\tau)\big)\,v(\tau)\,d\tau\ \ge\ 0,
\]
since $\nabla\sigma(\cdot)\succeq 0$, and
\[
I_2(a)=\frac{1}{2\pi}\int_{0}^{2\pi} v'(\tau)^\top \nabla\sigma\!\big(av(\tau)\big)\,v(\tau)\,d\tau=0.
\]
To see $I_2(a)=0$, compute
\[
\frac{d}{d\tau}\big[v(\tau)^\top \sigma\!\big(av(\tau)\big)\big]
= v'(\tau)^\top \sigma\!\big(av(\tau)\big) + a\,v(\tau)^\top \nabla\sigma\!\big(av(\tau)\big)\,v'(\tau).
\]
Thus,
\begin{equation*}
    \begin{aligned}
     I_2(a)&=\frac{1}{2\pi a}\int_{0}^{2\pi}\frac{d}{d\tau}\big[v(\tau)^\top \sigma\!\big(av(\tau)\big)\big]\,d\tau \\
&-\frac{1}{2\pi a}\int_{0}^{2\pi} v'(\tau)^\top \sigma\!\big(av(\tau)\big)\,d\tau.   
    \end{aligned}
\end{equation*}
Since $v(\tau)^\top \sigma\!\big(av(\tau)\big)$ is $2\pi$-periodic, the first term vanishes,  and
\begin{equation*}
    \begin{aligned}
        \int_{0}^{2\pi} v'(\tau)^\top \sigma\big(av(\tau)\big)\,d\tau
&=\frac{1}{a}\int_{0}^{2\pi}\frac{d}{d\tau}\,\mathrm{lse}\big(a v(\tau)\big)\,d\tau\\
&=\frac{1}{a}\big[\mathrm{lse}\!\big(a v(\tau)\big)\big]_{0}^{2\pi}=0,
    \end{aligned}
\end{equation*}
 Hence $I_2(a)=0$ and
\[
\frac{\partial J_g}{\partial a}(\phi,a)=\cos\phi\,I_1(a)\ \ge\ 0,
\]
with strict positivity when $\cos\phi>0$ and $v(\cdot)\not\notin \mathrm{span}\{\mathbf{1}_n\}$, which holds if at least one of $\bar p$ or $\bar q$ is not in $\mathrm{span}\{\mathbf{1}_n\}$.
\end{proof}
The gain–monotonicity result above immediately yields a performance comparison between learning dynamics that share the same phase lag at a given frequency but differ in gain.
\begin{corollary}
\label{cor:gain_dominance_fixed_phase}
Fix $\omega>0$ and apply the payoff
\[
p(t)=\bar p\,\sin(\omega t)+\bar q\,\cos(\omega t),\qquad \bar p,\bar q\in\mathbb{R}^n.
\]
Consider two learning dynamics $\sigma\circ g_1(s)I_n$ and $\sigma\circ g_2(s)I_n$ with
\begin{equation*}
    \begin{aligned}
        \phi&=\;\arg g_1(j\omega)=\arg g_2(j\omega), \\
        a_k&=|g_k(j\omega)|>0,\ k\in\{1,2\}.
    \end{aligned}
\end{equation*}
If $a_1>a_2$ and $\cos\phi\ge 0$, then
\[
J_{g_1}\ \ge\ J_{g_2}.
\]
Moreover, if $\cos\phi>0$ and at least one of $\bar p$ or $\bar q$ is not in $\mathrm{span}\{\mathbf{1}_n\}$, then the inequality is strict: $J_{g_1}>J_{g_2}$.
\end{corollary}

\begin{proof}
By Proposition~\ref{prop:gain_monotonicity_fixed_phase}, for fixed phase $\phi$ with $\cos\phi\ge 0$, the map $a\mapsto J_g(\phi,a)$ is nondecreasing (and strictly increasing when $\cos\phi>0$ and $\bar p$ or $\bar q$ has a nonzero component in $T\Delta_n$). Since $a_1>a_2$ and both models share the same $\phi$, the claims follow.
\end{proof}

\subsection{Summary and Visualization}

The preceding results quantify how the asymptotic average reward of learning
dynamics of the form $\sigma\circ g(s)I_n$ depends on the phase lag
$\phi=-\arg(g(j\omega))$ and gain $a=|g(j\omega)|$ of the transfer function
$g(s)$ under sinusoidal payoff environments.

Figure~\ref{fig:asym_avg_reward_phi_a} visualizes the asymptotic average reward \(J(\phi,a)\) attained by learning dynamics of the form \(\sigma\circ g(s)I_n\) in the environment \(p(t)=\bar p\,\sin(\omega t)+\bar q\,\cos(\omega t)\).  The horizontal axis is the phase lag \(\phi=-\arg\big(g(j\omega)\big)\in[-{\pi}/{2},{\pi}/{2}]\) for passive $g(s)$, while each colored curve corresponds to a different gain \(a=\lvert g(j\omega)\rvert\) (with the colorbar indicating increasing \(a\), blue\(\to\)red). 
To read off the predicted performance of a particular model at frequency \(\omega\), compute \(\phi\) and \(a\) from \(g(j\omega)\), then evaluate \(J(\phi,a)\) on the corresponding curve. 
The plot is generated with \(\bar p=\begin{bmatrix}1 & -0.4 & 0.6\end{bmatrix}^{\!\top}\) and \(\bar q=\begin{bmatrix}0.5 & 0.1 & -0.6\end{bmatrix}^{\!\top}\).

\begin{figure}[H]
    \centering
    \includegraphics[width=0.8\linewidth]{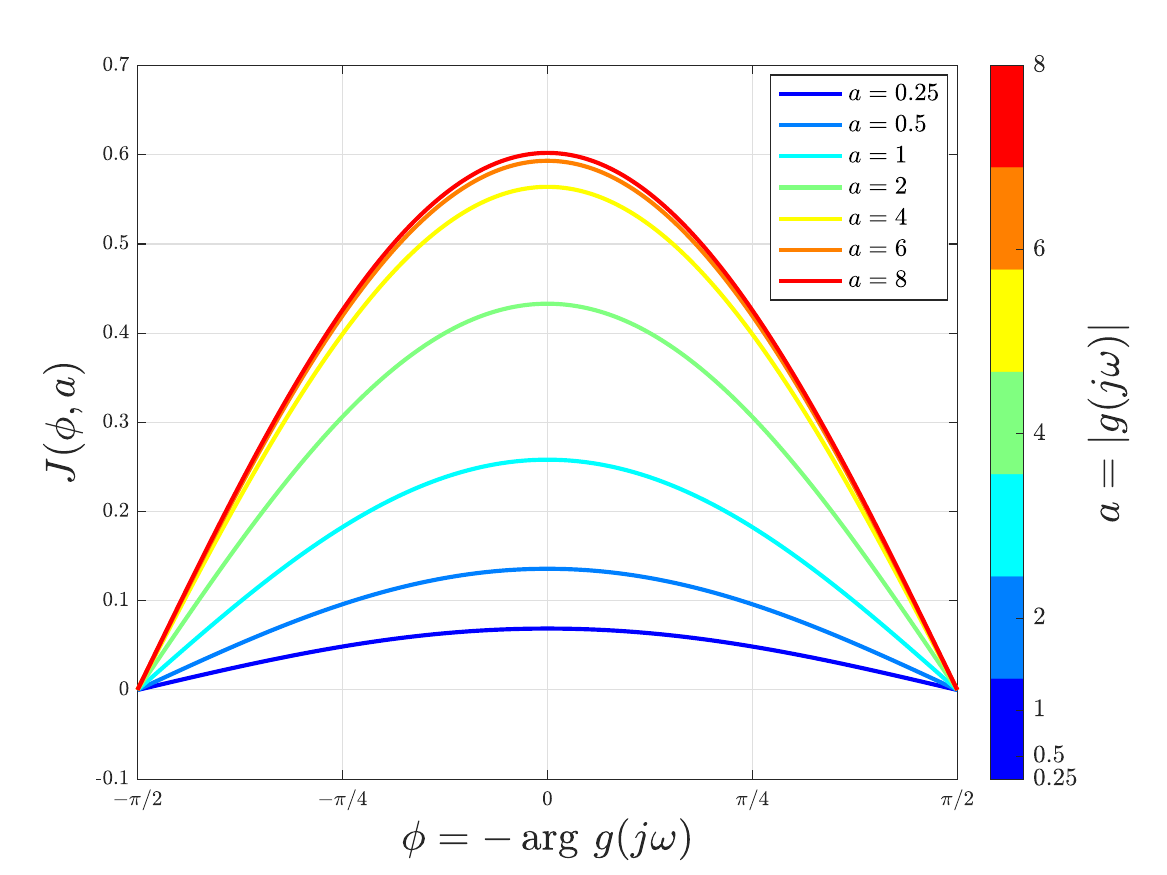}
    \caption{Asymptotic average reward for different values of $\phi=-\arg(g(j\omega))$ and $a=|g(j\omega)|$.}
    \label{fig:asym_avg_reward_phi_a}
\end{figure}

The analysis above studies the effect of  phase lag and gain,  associated with $g(s)$, on the asymptotic performance of the learning model of the form $\sigma \circ g(s)I_n$ under sinusoidal payoff environments. A natural next question is how these insights extend to general payoff trajectories $p(\cdot)$ that are not necessarily sinusoidal. In the next section, we address this question by comparing the performance of the exponential replicator dynamics (Ex--RD, \eqref{eq:EX-RD}) with its predictive variant employing a low--pass filter.
%%%%%%%%%%%%%%%%%%%%%%%%%%%%%%%%%%%%%%%%%%%%%%%%%%%%%%%%%%%%%%%%%%
\section{Performance comparison between Ex-RD and predictive Ex--RD}\label{sec:predictive_Ex_RD_and_Ex_RD}
\subsection{Ex--RD and predictive Ex--RD}

In this section, we compare the performance of the exponential replicator
dynamics (Ex--RD) with its predictive variant. For reference, Ex--RD with
$\lambda=1$ is characterized by
\begin{equation}
\begin{aligned}
\dot{z}&= p - z,\\
x &= \sigma(z).
\end{aligned}
\label{eq:Ex-RD_lambda=1}
\end{equation}
In the predictive Ex--RD, we employ a first--order low–pass predictor; the resulting dynamics are
\begin{equation}
\begin{aligned}
\dot{r} &= p - r,\\
\dot{m} &= p - m,\\
z &= r + m,\\
x &= \sigma(z).
\end{aligned}
\label{eq:predictive_Ex_RD}
\end{equation}

\subsection{Uniform Dominance of Predictive Ex--RD}

The next result shows that predictive Ex--RD uniformly dominates Ex--RD for
any payoff trajectory.

\begin{theorem}
\label{thm:pred-exrd-dominates}
The predictive Ex--RD \eqref{eq:predictive_Ex_RD} uniformly dominates the standard Ex--RD \eqref{eq:Ex-RD_lambda=1}. Specifically, for any payoff trajectory $p(\cdot)$, let $x_{\mathrm{E}}(\cdot)$ denote the strategy generated by Ex--RD and $x_{\mathrm{PE}}(\cdot)$ the strategy generated by predictive Ex--RD. Then,
for all $T>0$,
\[
\int_{0}^{T} p(t)^\top x_{\mathrm{PE}}(t)\,dt
\;\ge\;
\int_{0}^{T} p(t)^\top x_{\mathrm{E}}(t)\,dt,
\]
assuming matched initial conditions $z_{\mathrm{E}}(0)=r(0)=m(0)=\alpha\,\mathbf{1}_n$ for some $\alpha\in\mathbb{R}$. Moreover, the inequality is strict for every $T>0$ unless $p(t)\in\operatorname{span}\{\mathbf{1}_n\}$ for all $t\ge 0$.
\end{theorem}

\begin{proof}
Let $z_{\mathrm{E}}$ denote the Ex-RD score and
$z_{\mathrm{PE}}$ denotes the predictive Ex-RD score. Since
\[
\dot z_{\mathrm{E}}=p-z_{\mathrm{E}},\qquad
\dot r=p-r,\qquad
\dot m=p-m,
\]
with matched initial conditions $z_{\mathrm{E}}(0)=r(0)=m(0)$, it follows that
$r(t)=m(t)=z_{\mathrm{E}}(t)$ for all $t\ge0$, and hence
\(
z_{\mathrm{PE}}(t)=2z_{\mathrm{E}}(t).
\)
Therefore, the cumulative–reward gap between the predictive Ex-RD and the standard Ex-RD is 
\[
\Delta P(T)
\;=\;
\int_{0}^{T} p(t)^{\top}
\bigl[\sigma \bigl(2\,z_{\mathrm{E}}(t)\bigr)
      -\sigma \bigl(z_{\mathrm{E}}(t)\bigr)\bigr]\,dt.
\]

By the mean–value theorem for vector–valued maps,
\[
\sigma(2z_{\mathrm{E}})-\sigma(z_{\mathrm{E}})
=\int_{0}^{1}\nabla\sigma\bigl((1+s)z_{\mathrm{E}}\bigr)\,z_{\mathrm{E}}\,ds,
\]
so
\[
\Delta P (T)
=\int_{0}^{1}
\underbrace{\int_{0}^{T}
p(t)^{\top}\nabla\sigma\bigl((1+s)z_{\mathrm{E}}(t)\bigr)\,z_{\mathrm{E}}(t)\,dt}_{Q(s)}
\,ds.
\]
It suffices to show that $Q(s)\ge0$ for all $s\in[0,1]$.

Fix $s\in[0,1]$ and define $\xi\coloneqq (1+s)z_{\mathrm{E}}$. Then
\begin{equation}
    \begin{aligned}
        \dot{\xi} & = (1+s) p - \xi \\
          y&=\nabla\sigma\bigl((1+s)z_{\mathrm{E}}\bigr)\,z_{\mathrm{E}}
=\frac{1}{1+s}\,\nabla\sigma(\xi)\,\xi.
    \end{aligned}
    \label{eq:nonlinear_sys_comp_ex}
\end{equation}

We claim passivity of the nonlinear system \eqref{eq:nonlinear_sys_comp_ex} from input $p$ to output $y$.
Consider the storage function
\[
V(\xi)\;\coloneqq\;\frac{1}{(1+s)^2}\,
D_{\mathrm{KL}}\!\Bigl(\sigma(\xi)\,\Big\|\,\tfrac{1}{n}\mathbf{1}_n\Bigr),
\]
which satisfies $V(\xi)\ge0$ and can be written explicitly as
\begin{equation*}
\begin{aligned}
V(\xi)
&=\frac{1}{(1+s)^2}
\sum_{i=1}^{n}\sigma(\xi)_i
\ln\!\left(\frac{\sigma(\xi)_i}{1/n}\right) \\
&=\frac{1}{(1+s)^2}
\Bigl(\ln n + \xi^\top\sigma(\xi)
      -\ln\!\sum_{j=1}^n e^{\xi_j}\Bigr).
\end{aligned}
\end{equation*}

Differentiating along trajectories yields
\[
\begin{aligned}
\dot V
&=\nabla_\xi V(\xi)^\top \dot\xi \\
&=\frac{1}{1+s}
\bigl(\sigma(\xi)+\nabla\sigma(\xi)\,\xi-\sigma(\xi)\bigr)^\top (p-z_{\mathrm{E}}) \\
&=\frac{1}{s+1}p^\top\nabla\sigma(\xi)\,\xi
   -z_{\mathrm{E}}^\top\nabla\sigma(\xi)\,z_{\mathrm{E}}
\;\le\; p^\top y,
\end{aligned}
\]
since $\nabla\sigma(\cdot)\succeq0$ implies
$z_{\mathrm{E}}^\top\nabla\sigma(\xi)\,z_{\mathrm{E}}\ge0$.

Integrating over $[0,T]$ gives
\[
\int_{0}^{T}
p(t)^\top\nabla\sigma\bigl((1+s)z_{\mathrm{E}}(t)\bigr)\,z_{\mathrm{E}}(t)\,dt
\;\ge\;
V\bigl(\xi(T)\bigr)-V\bigl(\xi(0)\bigr).
\]
Under the initialization $z_{\mathrm{E}}(0)=\alpha\,\mathbf{1}_n$ for some $\alpha\in\mathbb{R}$, all score components are initialized equally and no channel is favored. Consequently,
$\xi(0)=(1+s)\alpha\,\mathbf{1}_n$, which implies
$\sigma(\xi(0))=\tfrac{1}{n}\mathbf{1}_n$ and hence $V(\xi(0))=0$.
It follows that $Q(s)\ge 0$ for all $s\in[0,1]$, and therefore
$\Delta P\ge 0$. Equivalently, predictive Ex--RD produces cumulative reward no smaller than that of Ex--RD on any horizon $[0,T]$, i.e., predictive Ex--RD uniformly dominates Ex--RD.

\medskip
\noindent\emph{Strictness.}
Since $\nabla\sigma(\cdot)\succeq 0$ and $\ker\!\bigl(\nabla\sigma(\cdot)\bigr)=\operatorname{span}\{\mathbf{1}_n\}$,
the integrand in $Q(s)$ is nonnegative and vanishes if and only if
$p(t)\in\operatorname{span}\{\mathbf{1}_n\}$ and
$z_{\mathrm{E}}(t)\in\operatorname{span}\{\mathbf{1}_n\}$.
From the Ex--RD dynamics $\dot z_{\mathrm{E}}=p-z_{\mathrm{E}}$, if
$p(t)=\beta(t)\mathbf{1}_n$ for all $t\ge 0$, then
$z_{\mathrm{E}}(t)=\gamma(t)\mathbf{1}_n$ for all $t\ge 0$, where
\[
\gamma(t)=\int_{0}^{t} e^{-(t-\tau)}\beta(\tau)\,d\tau.
\]
Conversely, if $p(t)$ has a nonzero component orthogonal to $\mathbf{1}_n$ on a
set of nonzero measure in time, then the integrand is strictly positive . Consequently, $Q(s)=0$ if and only if
$p(t)\in\operatorname{span}\{\mathbf{1}_n\}$ for all $t\ge 0$, in which case
$\Delta P=0$; otherwise, $\Delta P>0$.
\end{proof}
\subsection{Frequency-Domain Interpretation}
The predictive Ex--RD \eqref{eq:predictive_RD_low_pass} and the standard Ex--RD \eqref{eq:Ex-RD_lambda=1} can both be represented as cascade interconnections between diagonal LTI systems $g(s)I_n$ and the static softmax mapping $\sigma(\cdot)$. Specifically, predictive Ex-RD corresponds to $g_{\mathrm{PE}}(s)=\tfrac{2}{s+1} $, while Ex-RD corresponds to $g_{\mathrm{E}}(s)=\tfrac{1}{s+1}$. For any frequency $\omega\in\mathbb{R}$, the two transfer functions have identical phase lag, $\phi=\arctan(\omega)$, but satisfy
$|g_{\mathrm{PE}}(j\omega)|=2|g_{\mathrm{E}}(j\omega)|$. By
Corollary~\ref{cor:gain_dominance_fixed_phase}, predictive Ex--RD therefore achieves a strictly larger asymptotic average reward than Ex-RD, in agreement with Theorem~\ref{thm:pred-exrd-dominates}.

While Theorem~\ref{thm:pred-exrd-dominates} establishes that one learning algorithm (predictive Ex--RD) can uniformly dominate another (standard Ex--RD) in any payoff environment, neither algorithm satisfies a finite-regret guarantee \cite{abdelraouf2025passivity}.

This raises the main fundamental question: \emph{can a learner regret not using a no-regret algorithm?} Equivalently, does there exist a no-regret learning dynamics that uniformly outperforms another no-regret dynamics across all environments?

To address this question, we turn next to the comparison between replicator dynamics and its higher-order variants. In particular, we study the \emph{anticipatory replicator dynamics}, which is equivalent to predictive replicator dynamics equipped with a first--order low--pass predictor (see Remark~\ref{rem:anticipatory_predicitive}).
%%%%%%%%%%%%%%%%%%%%%%%%%%%%%%%%%%%%%%%%%%%%%%%%%%%%%%%%%%%%
\section{Performance Comparison Between Replicator Dynamics and Its Higher-Order Variants}\label{sec:predictive_RD_and_RD}

In this section, we compare the performance of standard replicator dynamics (RD) with its higher-order, payoff--based variants. Our focus is on the anticipatory replicator dynamics \eqref{eq:anticipatory_RD}. We establish three results: (i) for constant payoffs $p(t)\equiv\bar p$, anticipatory RD uniformly dominates RD; (ii) locally, near an equilibrium, anticipatory RD uniformly dominates RD, and this result extends to any predictive variant with a passive asymptotically stable predictor; and
(iii) by formulating the global comparison as an optimal--control problem, we show that anticipatory RD uniformly dominates RD along arbitrary payoff trajectories.

Throughout this section, we denote by $(z_{\mathrm{R}},x_{\mathrm{R}})$ the score and strategy generated by the standard RD, and by
$(z_{\mathrm{PR}},x_{\mathrm{PR}})$ the corresponding quantities generated by the predictive (anticipatory) replicator dynamics. 

\subsection{Fixed-payoff performance comparison}
\label{subsec:performance_comparison_fixed}
We begin with the simplest setting in which the payoff is constant.
\begin{proposition}
\label{prop:fixed_payoff_PRD}
Let $p(t)\equiv\bar p\in\mathbb{R}^n$ be constant. The predictive replicator
dynamics with a low--pass predictor \eqref{eq:predictive_RD_low_pass} uniformly
dominates the standard replicator dynamics \eqref{eq:RD}, i.e.,
\[
\int_{0}^{T} \bar p^\top x_{\mathrm{PR}}(t)\,dt
\;\ge\;
\int_{0}^{T} \bar p^\top x_{\mathrm{R}}(t)\,dt
\qquad \text{for all } T>0,
\]
assuming matched initial conditions $r(0)=z_{\mathrm{R}}(0)$ and
$m(0)=\mathbf{0}_n$. Moreover, the inequality is strict if and only if
$\bar p\notin \operatorname{span}\{\mathbf{1}_n\}$.
\end{proposition}

\begin{proof}
For standard RD, the score satisfies
\[
z_{\mathrm{R}}(t)=z_{\mathrm{R}}(0)+\int_{0}^{t} \bar p\,d\tau
= z_{\mathrm{R}}(0)+t\,\bar p.
\]
For the predictive replicator dynamics with a low--pass predictor,
\[
m(t)=e^{-\lambda t}m(0)
+\int_{0}^{t} e^{-\lambda(t-\tau)}\gamma\lambda\,\bar p\,d\tau
=\gamma\bigl(1-e^{-\lambda t}\bigr)\bar p,
\]
where we used $m(0)=\mathbf{0}_n$. With the matched initialization
$r(0)=z_{\mathrm{R}}(0)$, it follows that
\[
z_{\mathrm{PR}}(t)-z_{\mathrm{R}}(t)=m(t)
=\gamma\bigl(1-e^{-\lambda t}\bigr)\bar p.
\]
 The cumulative reward difference at any $T>0$ is defined as 
\[
\Delta P(T)
\;\coloneqq\;
\int_{0}^{T}
\bar p^\top\!\bigl(\sigma(z_{\mathrm{PR}}(t))
-\sigma(z_{\mathrm{R}}(t))\bigr)\,dt.
\]
By the integral mean--value theorem for vector--valued maps,
\[
\sigma(z_{\mathrm{PR}})-\sigma(z_{\mathrm{R}})
=\int_{0}^{1}
\nabla \sigma(\xi)
\,(z_{\mathrm{PR}}-z_{\mathrm{R}})\,ds.
\]
where $\xi = z_{\mathrm{R}}+s(z_{\mathrm{PR}}-z_{\mathrm{R}})$. Then, substituting into  $\Delta P$ yields
\[
\Delta P(T)
=\int_{0}^{T}\!\int_{0}^{1}
\gamma\bigl(1-e^{-\lambda t}\bigr)
\,\bar p^\top
\nabla\sigma\bigl(\xi\bigr)
\,\bar p
\,ds\,dt.
\]

Since $\nabla\sigma(\cdot)\succeq 0$ with
$\ker\nabla\sigma(\cdot)=\operatorname{span}\{\mathbf{1}_n\}$ and $\gamma(1-e^{-\lambda t})>0$ for all $t>0$, the integrand is
nonnegative for all $(t,s)$, and strictly positive whenever
$\bar p\notin\operatorname{span}\{\mathbf{1}_n\}$. Hence $\Delta P\ge 0$, with strict inequality if and only if $\bar p\notin\operatorname{span}\{\mathbf{1}_n\}$. This proves the claim.
\end{proof}

\subsection{Local performance comparison via a passivity viewpoint}
\label{subsec:performance_comparison_local}

Proposition~\ref{prop:fixed_payoff_PRD} establishes that anticipatory RD uniformly dominates standard RD in the case of a constant payoff vector. Our broader objective, however, is to compare the performance of the two learning dynamics under \emph{arbitrary} payoff trajectories.

To this end, we introduce a comparison system whose input is the payoff signal
$p$ and whose output is the difference between the strategies generated by the
anticipatory and standard replicator dynamics. This system is described by
\begin{equation}
\label{eq:nonlinear_sys_comparison}
\begin{aligned}
\dot z_{\mathrm{R}} &= p,\\
\dot m &= \gamma\lambda\,p - \lambda m,\\
y &= \sigma(z_{\mathrm{R}}+m)-\sigma(z_{\mathrm{R}}),
\end{aligned}
\end{equation}
where $y = x_{\mathrm{PR}}-x_{\mathrm{R}}$ denotes the strategy difference.

If the comparison system \eqref{eq:nonlinear_sys_comparison} is passive from the
input $p$ to the output $y$, then there exists a storage function $V$ such that
\small{
\[
\int_{0}^{T} p(t)^\top
\bigl(\underbrace{\sigma(z_{\mathrm{R}}(t)+m(t))}_{x_{\mathrm{PR}}(t)}
      -\underbrace{\sigma(z_{\mathrm{R}}(t))}_{x_{\mathrm{R}}(t)}\bigr)\,dt
\;\ge\; -V\bigl(z_{\mathrm{R}}(0),m(0)\bigr)
\]}
\normalsize
\noindent for all $T>0$. Under unbiased initialization, i.e.,
$z_{\mathrm{R}}(0)=m(0)=\mathbf{0}_n$, one has $V(\mathbf{0}_n,\mathbf{0}_n)=0$,
and passivity of the comparison system implies that anticipatory RD uniformly
dominates standard RD for any payoff trajectory.

Thus, the performance comparison problem can be reformulated as a passivity
analysis of the nonlinear system \eqref{eq:nonlinear_sys_comparison}. The
presence of the softmax nonlinearity in the output equation, however, makes a
global passivity analysis challenging. We therefore proceed by studying the
passivity of the \emph{linearized} comparison system about an equilibrium
trajectory and translating this local passivity property into a local
performance dominance result.

\begin{theorem}
\label{thm:local_dominance_uniform}
Consider the comparison system \eqref{eq:nonlinear_sys_comparison} and the
operating trajectory
\[
p^*(t)=\mathbf{1}_n,\quad
m^*(t)=\gamma\,\mathbf{1}_n,\quad
z_{\mathrm{R}}^*(t)=t\,\mathbf{1}_n,\quad
y^*(t)=\mathbf{0}_n.
\]
Let $p(t)=p^*(t)+\Delta p(t)$ and $y(t)=y^*(t)+\Delta y(t)$ denote the input and
output deviations. Then the linearization of
\eqref{eq:nonlinear_sys_comparison} about
$\bigl(p^*,m^*,z_{\mathrm{R}}^*\bigr)$ is passive. Consequently, anticipatory RD
uniformly dominates standard RD locally, i.e., for all sufficiently small
$\Delta p$,
\[
\int_{0}^{T}
\bigl(p^*(t)+\Delta p(t)\bigr)^{\!\top}
\bigl(y^*(t)+\Delta y(t)\bigr)\,dt
\;\ge\; 0
\qquad \forall T>0.
\]
\end{theorem}

\begin{proof}
We establish local passivity of the comparison system
\eqref{eq:nonlinear_sys_comparison} about a nominal trajectory and then translate this property into a local reward dominance result.

We linearize the comparison system \eqref{eq:nonlinear_sys_comparison} about the
uniform operating trajectory $(p^*,m^*,z_{\mathrm{R}}^*)$ s. Along this trajectory,
\[
x_{\mathrm{R}}^*(t)=x_{\mathrm{PR}}^*(t)=\tfrac{1}{n}\mathbf{1}_n,
\qquad
y^*(t)=\mathbf{0}_n.
\]

Let
\(
\Delta z_{\mathrm{R}}=z_{\mathrm{R}}-z_{\mathrm{R}}^*,\;
\Delta m=m-m^*,\;
\Delta p=p-p^*,\;
\Delta y=y-y^*
\)
denote deviations. Linearizing
\eqref{eq:nonlinear_sys_comparison} yields
\[
\dot{\Delta z_{\mathrm{R}}}=\Delta p,\quad
\dot{\Delta m}=\gamma\lambda\,\Delta p-\lambda\,\Delta m,\quad
\Delta y=\Delta x_{\mathrm{PR}}-\Delta x_{\mathrm{R}},
\]
where
\[
\Delta x_{\mathrm{PR}}
=\nabla\sigma(z_{\mathrm{R}}^*+m^*)(\Delta z_{\mathrm{R}}+\Delta m),
\quad
\Delta x_{\mathrm{R}}
=\nabla\sigma(z_{\mathrm{R}}^*)\,\Delta z_{\mathrm{R}}.
\]

Using the identity
\(
\nabla\sigma(v)=\mathrm{diag}(\sigma(v))-\sigma(v)\sigma(v)^\top
\)
and the translation invariance
\(
\sigma(v+\alpha\mathbf{1}_n)=\sigma(v)
\),
the Jacobian along the nominal trajectory is constant:
\[
S \;\coloneqq\;
\nabla\sigma(z_{\mathrm{R}}^*)
=\nabla\sigma(z_{\mathrm{R}}^*+m^*)
=\frac{1}{n}I_n-\frac{1}{n^2}\mathbf{1}_n\mathbf{1}_n^\top.
\]
Hence,
\[
\Delta x_{\mathrm{R}}=S\,\Delta z_{\mathrm{R}},\quad
\Delta x_{\mathrm{PR}}=S(\Delta z_{\mathrm{R}}+\Delta m),\quad
\Delta y=S\,\Delta m.
\]

\paragraph*{Tangent–space reduction.}
Admissible deviations of strategies lie in the tangent space of the simplex,
\(
T\Delta_n.
\)
Let $N\in\mathbb{R}^{n\times(n-1)}$ have orthonormal columns spanning $T\Delta_n$,
so that $\mathbf{1}_n^\top N=0$ and $N^\top N=I_{n-1}$. We express deviations as
\[
\Delta z_{\mathrm{R}}=N\delta z,\quad
\Delta m=N\delta m,\quad
\Delta p=N\delta p,\quad
\Delta y=N\delta y.
\]
Since $SN=\tfrac{1}{n}N$, the reduced dynamics become
\[
\dot{\delta m}=\gamma\lambda\,\delta p-\lambda\,\delta m,
\qquad
\delta y=\frac{1}{n}\,\delta m.
\]
Thus, the linearized map from $\delta p$ to $\delta y$ is an LTI system with
transfer function
\[
\frac{\gamma\lambda}{n}\,\frac{1}{s+\lambda}I_{n-1},
\]
which is strictly passive for $\gamma,\lambda>0$. Define the quadratic storage function
\[
V(\delta m)\;\coloneqq\;
\frac{1}{2n\gamma\lambda}\,\|\delta m\|^2.
\]
Along trajectories,
\[
\dot V
=\frac{1}{n}\,\delta m^\top\delta p
-\frac{1}{n\gamma}\,\|\delta m\|^2
\;\le\;
\delta p^\top\delta y,
\]
since $\delta p^\top\delta y=\tfrac{1}{n}\delta m^\top\delta p$. Integrating over
$[0,T]$ yields the passivity inequality
\[
\int_0^T \delta p(t)^\top\delta y(t)\,dt
\;\ge\;
V(\delta m(T))-V(\delta m(0)).
\]
For zero initial deviations, $\delta m(0)=0$, we obtain
\[
\int_0^T \delta p^\top\delta y\,dt \;\ge\; 0,
\]
with strict inequality for any nontrivial perturbation $\delta p$, due to the
strict dissipation term.

\paragraph*{Translation to reward dominance.}
Writing
\(
p=\mathbf{1}_n+N\delta p
\),
\(
x_{\mathrm{R}}=\tfrac{1}{n}\mathbf{1}_n+N\delta x_{\mathrm{R}}
\),
and
\(
x_{\mathrm{PR}}=\tfrac{1}{n}\mathbf{1}_n+N\delta x_{\mathrm{PR}}
\),
the reward gap satisfies
\[
\begin{aligned}
\Delta P
&=\int_0^T p^\top(x_{\mathrm{PR}}-x_{\mathrm{R}})\,dt
=\int_0^T \delta p^\top\delta y\,dt
\;\ge\;0,
\end{aligned}
\]
with strict inequality for any nonzero $\delta p$. This proves local uniform
dominance.

\paragraph*{Small–signal boundedness.}
For the nonlinear system with $p=\mathbf{1}_n+\Delta p$, the deviation dynamics satisfy
\begin{align*}
\dot{\Delta z_{\mathrm{R}}}&=\Delta p,\\
\dot{\Delta m}&=\gamma\lambda\,\Delta p-\lambda\,\Delta m,\\
\Delta y&=\sigma(\Delta z_{\mathrm{R}}+\Delta m)-\sigma(\Delta z_{\mathrm{R}}).
\end{align*}
By Lemma~\ref{lem:softmax_lipschitz},
\(
\|\Delta y(t)\|_2\le\tfrac{1}{2}\|\Delta m(t)\|_2
\)
pointwise. Moreover, the map from $\Delta p$ to $\Delta m$ has transfer function
$g(s)=\tfrac{\gamma\lambda}{s+\lambda}$ with $\|g\|_\infty=\gamma$. Consequently,
\[
\|\Delta y\|_{2,[0,T]}
\;\le\;\frac{\gamma}{2}\,\|\Delta p\|_{2,[0,T]},
\]
showing that sufficiently small payoff deviations produce proportionally small
output deviations. This justifies the local passivity–based dominance result.
\end{proof}
The local dominance result established above for predictive RD with a first--order low--pass predictor (anticipatory RD) can be extended to predictive RD equipped with an arbitrary passive, asymptotically stable predictor. In this more general setting, the prediction mechanism is described by a stable linear system driven by the payoff signal. Specifically, consider the comparison system
\begin{equation}
\label{eq:nonlinear_genearalized_comparison}
\begin{aligned}
\dot z_{\mathrm{R}} &= p,\\
\dot x_h &= (A_h\!\otimes\! I_n)\,x_h + (B_h\!\otimes\! I_n)\,p,\\
m &= (C_h\!\otimes\! I_n)\,x_h,\\
y &= \sigma(z_{\mathrm{R}}+m)-\sigma(z_{\mathrm{R}}),
\end{aligned}
\end{equation}
where $x_h\in\mathbb{R}^{mn}$ denotes the predictor state and
$h(s)=C_h(sI_m-A_h)^{-1}B_h$ is the associated prediction transfer function.

\begin{theorem}
\label{thm:local_dominance_general_predictor}
Let the prediction transfer function
\(
h(s)=C_h(sI_m-A_h)^{-1}B_h
\) be passive and asymptotically stable. Then the generalized predictive RD \eqref{eq:predictive_RD} uniformly dominates the standard RD locally, in the sense that the comparison system
\eqref{eq:nonlinear_genearalized_comparison} is locally passive from the payoff input $p$ to the strategy difference output $y$.
\end{theorem}
\begin{proof}
The argument follows the same steps as the proof of
Theorem~\ref{thm:local_dominance_uniform}. We again view the performance gap as a passivity of the comparison system \eqref{eq:nonlinear_genearalized_comparison}. Let
\(
g(0)\;\coloneqq\;C_h(-A_h)^{-1}B_h
\)
denote the DC gain of the predictor, and consider the uniform operating
trajectory
\begin{align*}
    p^*(t)&=\mathbf{1}_n,\qquad \;\; \; x_h^*=\bigl((-A_h)^{-1}B_h\bigr)\!\otimes\!\mathbf{1}_n,\quad \\
    m^*&=g(0)\,\mathbf{1}_n,\quad
   z_{\mathrm{R}}^*(t)=t\,\mathbf{1}_n,
\end{align*}

for which
$x_{\mathrm{R}}^*(t)=x_{\mathrm{PR}}^*(t)=\tfrac{1}{n}\mathbf{1}_n$ and
$y^*(t)=\mathbf{0}_n$.

Linearizing \eqref{eq:nonlinear_genearalized_comparison} about this trajectory yields a deviation system of the same structure as in
Theorem~\ref{thm:local_dominance_uniform}, except that the low--pass predictor is replaced by the LTI system with state $\Delta x_h$:
\begin{align*}
    \dot{\Delta x_h}&=(A_h\otimes  I_n)\Delta x_h+(B_h \otimes I_n)\Delta p,\\
    \Delta m&=(C_h\otimes I_n)\Delta x_h,\\
     \Delta y &= S \Delta m.
\end{align*}
where $S=\frac{1}{n}I_n-\frac{1}{n^2}\mathbf{1}_n\mathbf{1}_n^\top$. As before, admissible perturbations lie in the tangent space
$T\Delta_n$. Projecting the dynamics onto this space
and using the translation invariance of the softmax, the reduced linearized
input--output map from $\delta p$ to $\delta y$ is $\frac{1}{n} h(s) I_{n-1}$ where $h(s)=C_h(sI_m-A_h)^{-1}B_h$.

By assumption, $g(s)$ is passive and asymptotically stable; hence the reduced
map is passive. Therefore, there exists a quadratic storage function
$V(\delta x_h)$ such that
\[
\int_0^T \delta p(t)^\top\delta y(t)\,dt
\;\ge\;
V(\delta x_h(T))-V(\delta x_h(0)).
\]
With zero initial deviation, this yields
$\int_0^T \delta p^\top\delta y\,dt\ge0$. 

Finally, as in the proof of
Theorem~\ref{thm:local_dominance_uniform}, the reward gap satisfies
\[
\int_0^T p^\top(x_{\mathrm{PR}}-x_{\mathrm{R}})\,dt
=\int_0^T \delta p^\top\delta y\,dt \;\ge\; 0,
\]
and is strictly positive if $h(s)$ is strictly passive.

Local boundedness of the nonlinear deviation dynamics follows directly from the Lipschitz continuity of the softmax (Lemma \ref{lem:softmax_lipschitz}) and the bounded $\mathcal{H}_\infty$ norm of the asymptotically stable predictor $h(s)$, completing the proof.
\end{proof}

\subsection{Global performance comparison via an optimal control formulation}
\label{subsec:performance_comparison_global}

Theorem~\ref{thm:local_dominance_general_predictor} establishes that, in a neighborhood of the uniform operating trajectory
$(p^*=\mathbf{1}_n,\,y^*=\mathbf{0}_n)$, predictive RD equipped with a passive, asymptotically stable predictor, including the first--order low--pass filter (anticipatory RD), uniformly dominates standard RD. We now move beyond this local comparison and study performance dominance at a global level.

In this subsection, we compare the performance of anticipatory RD and  standard RD globally, for arbitrary payoff trajectories. We formulate the cumulative reward gap as an \emph{optimal control problem} and apply Pontryagin’s Maximum Principle to characterize the worst--case (i.e., minimal) performance difference. This analysis yields a global dominance result: anticipatory RD uniformly dominates standard RD across all payoff environments.

\begin{theorem}
Anticipatory RD \eqref{eq:predictive_RD_low_pass} uniformly dominate the standard RD \eqref{eq:RD} for every payoff trajectory. Equivalently,
\[
J(p)\;\coloneqq\; \int_0^T p(t)^\top \bigl(x_{\mathrm{PR}}(t)-x_{\mathrm{R}}(t)\bigr)\,dt \;\ge\; 0
\qquad \forall\,T>0,
\]
under the matched initial conditions \(r(0)=z_{\mathrm{R}}(0)=\mathbf{0}_n\).
\end{theorem}

\begin{proof}
    With \(r(0)=z_{\mathrm{R}}(0)\), we have \(r(t)=z_{\mathrm{R}}(t)\) for all \(t\), so the comparison dynamics reduce to
    \begin{equation*}
        \begin{aligned}
      \dot{z}_{\mathrm{R}} &= p,\\
      \dot{m} &= \gamma \lambda\, p - \lambda\, m,\\
      y &= \sigma(z_{\mathrm{R}}+m)-\sigma(z_{\mathrm{R}}) \\
        \end{aligned}
    \end{equation*}
    The cumulative reward difference between anticipatory RD and standard RD is
\[
J(p)
=\int_0^T
p(t)^\top\bigl(\sigma(z_{\mathrm{R}}(t)+m(t))-\sigma(z_{\mathrm{R}}(t))\bigr)\,dt.
\]
To characterize the worst--case performance gap, we seek the payoff trajectory
$p(\cdot)$ that minimizes $J(p)$. This leads to the optimal control problem
\[
\begin{aligned}
\min_{p(\cdot)}\quad
&J(p)
=\int_{0}^{T}
p(t)^\top\bigl(\sigma(z_{\mathrm{R}}(t)+m(t))-\sigma(z_{\mathrm{R}}(t))\bigr)\,dt,\\
\text{s.t.}\quad
&\dot z_{\mathrm{R}}(t)=p(t),\qquad z_{\mathrm{R}}(0)=\mathbf{0}_n,\\
&\dot m(t)=\gamma\lambda\,p(t)-\lambda\,m(t),\qquad m(0)=\mathbf{0}_n,
\end{aligned}
\]
with free terminal state. For notational brevity, we suppress the explicit time
dependence below.

Let $p^*$ denote an optimal payoff trajectory with corresponding state
$(z_{\mathrm{R}}^*,m^*)$. Define the Hamiltonian
\[
\begin{aligned}
H(p,z_{\mathrm{R}},m,\psi_1,\psi_2)
&=\psi_1^\top p+\psi_2^\top(\gamma\lambda\,p-\lambda\,m)\\
&\quad -p^\top\bigl(\sigma(z_{\mathrm{R}}+m)-\sigma(z_{\mathrm{R}})\bigr),
\end{aligned}
\]
where $\psi_1,\psi_2$ are the costate variables. By Pontryagin’s Maximum
Principle, the optimal quadruple
$(p^*,z_{\mathrm{R}}^*,m^*,\psi_1^*,\psi_2^*)$ satisfies the canonical equations
\begin{equation}
\label{eq:canonical_equations_global}
\begin{aligned}
\dot z_{\mathrm{R}}^* &= p^*,\\
\dot m^* &= \gamma\lambda\,p^*-\lambda\,m^*,\\
\dot\psi_1^*
&=\nabla\sigma(z_{\mathrm{R}}^*+m^*)\,p^*
-\nabla\sigma(z_{\mathrm{R}}^*)\,p^*,\\
\dot\psi_2^*
&=\nabla\sigma(z_{\mathrm{R}}^*+m^*)\,p^*
+\lambda\,\psi_2^*,
\end{aligned}
\end{equation}
with transversality conditions $\psi_1^*(T)=\psi_2^*(T)=\mathbf{0}_n$. Since the Hamiltonian is affine in $p$, finiteness of the minimizing control
requires that the switching function vanish identically:
\[
\phi(t)
=\psi_1^*+\gamma\lambda\,\psi_2^*
-\sigma(z_{\mathrm{R}}^*+m^*)+\sigma(z_{\mathrm{R}}^*)
=\mathbf{0}_n,
\quad \forall t.
\]
Thus, the optimal control lies on a singular arc. Along the optimal trajectory, the Hamiltonian satisfies:
\begin{equation}\label{eq:Hzero}
H\bigl(p^*,z_{\mathrm{R}}^*,m^*,\psi_1^*,\psi_2^*\bigr)=0
\quad\text{for all }t\in[0,T].
\end{equation}
Differentiating $\phi(t)\equiv0$ and using
\eqref{eq:canonical_equations_global} yields
\[
\psi_2^*
=-\frac{1}{\gamma\lambda}\,
\nabla\sigma(z_{\mathrm{R}}^*+m^*)\,m^*.
\]
Using \eqref{eq:Hzero} and \(\phi\equiv 0\),
\[
0 \;=\; {p^*}^\top\underbrace{\bigl(\psi_1^* + \gamma\lambda\,\psi_2^*
- \sigma(z_{\mathrm{R}}^*+m^*) + \sigma(z_{\mathrm{R}}^*)\bigr)}_{=\;0}
- \lambda\,{\psi_2^*}^\top m^*
\]
 implies that \({\psi_2^*}^\top m^* = 0\). Substitute \(\psi_2^*\) to get
\[
m^{*\top}\,\nabla\sigma(z_{\mathrm{R}}^*+m^*)\,m^* \;=\; 0
\quad\text{for all }t>0.
\]
Since $\nabla\sigma(\cdot)\succeq0$ and
$\ker\nabla\sigma(u)=\operatorname{span}\{\mathbf{1}_n\}$, this equality holds if
and only if
\[
m^*(t)=\gamma(t)\,\mathbf{1}_n
\quad \forall t.
\]
From $\dot m^*=\gamma\lambda\,p^*-\lambda m^*$, it follows that
\[
p^*(t)
=\Bigl(\frac{\gamma(t)}{\gamma}
+\frac{\dot\gamma(t)}{\gamma\lambda}\Bigr)\mathbf{1}_n
\in\operatorname{span}\{\mathbf{1}_n\},
\quad \forall t.
\]
Therefore, the unique payoff trajectory minimizing $J(p)$ is one that lies
entirely in $\operatorname{span}\{\mathbf{1}_n\}$. For such trajectories,
\[
\begin{aligned}
J(p^*)
&=\int_0^T
{p^*}(t)^\top
\bigl(\sigma(z_{\mathrm{R}}^*(t)+\gamma(t)\mathbf{1}_n)
-\sigma(z_{\mathrm{R}}^*(t))\bigr)\,dt=0,
\end{aligned}
\]
since $\sigma(v+\alpha\mathbf{1}_n)=\sigma(v)$.
By Proposition~\ref{prop:fixed_payoff_PRD}, anticipatory RD uniformly dominates
standard RD for any fixed payoff. Hence, the global minimum of $J(p)$ over all
admissible payoff trajectories equals zero, implying
\[
J(p)\ge0
\quad \text{for all payoff trajectories } p(\cdot).
\]
This establishes global uniform dominance.
\end{proof}

The global dominance result is fully consistent with the preceding analyses.
First, we showed that for any fixed payoff $p(t)\equiv\bar p$ with
$\bar p\notin\operatorname{span}\{\mathbf{1}_n\}$, anticipatory RD \emph{strictly} and \emph{uniformly} dominates standard RD. Second, in
Subsection~\ref{subsec:performance_comparison_local}, we established a local
dominance result in a neighborhood of the uniform operating trajectory
$(p^*=\mathbf{1}_n,\;
x_{\mathrm{R}}^*=\tfrac{1}{n}\mathbf{1}_n,\;
x_{\mathrm{PR}}^*=\tfrac{1}{n}\mathbf{1}_n)$.
The global optimal--control analysis above unifies these findings by showing
that no payoff trajectory can produce a negative cumulative reward gap, thereby
extending both the fixed--payoff and local dominance results to arbitrary payoff
environments.

\subsection{Numerical experiment}

\begin{example}\label{ex:global_performance_comparison}
Consider the payoff environment
\[
p(t)
= v_1 \sin(t) + v_2 \cos(2t) + v_3 \sin(3t),
\]
where
\(
v_1 = [0.3,\;0.5,\;-0.7,\;1,\;0.8]^\top,\quad
v_2 = [0.7,\;2,\;0.4,\;1.2,\;2]^\top,\quad
v_3 = [0.8,\;1.4,\;-2.1,\;2,\;0.8]^\top .
\)
Figure~\ref{fig:global_comparison_RD_anti_RD} illustrates the cumulative average reward obtained by the anticipatory RD and the standard RD. As predicted by the global dominance result, the anticipatory RD achieves a strictly larger accumulated reward than the standard RD for all horizons $T>0$.
\end{example}

\begin{figure}[H]
    \centering
    \includegraphics[width=0.7\linewidth]{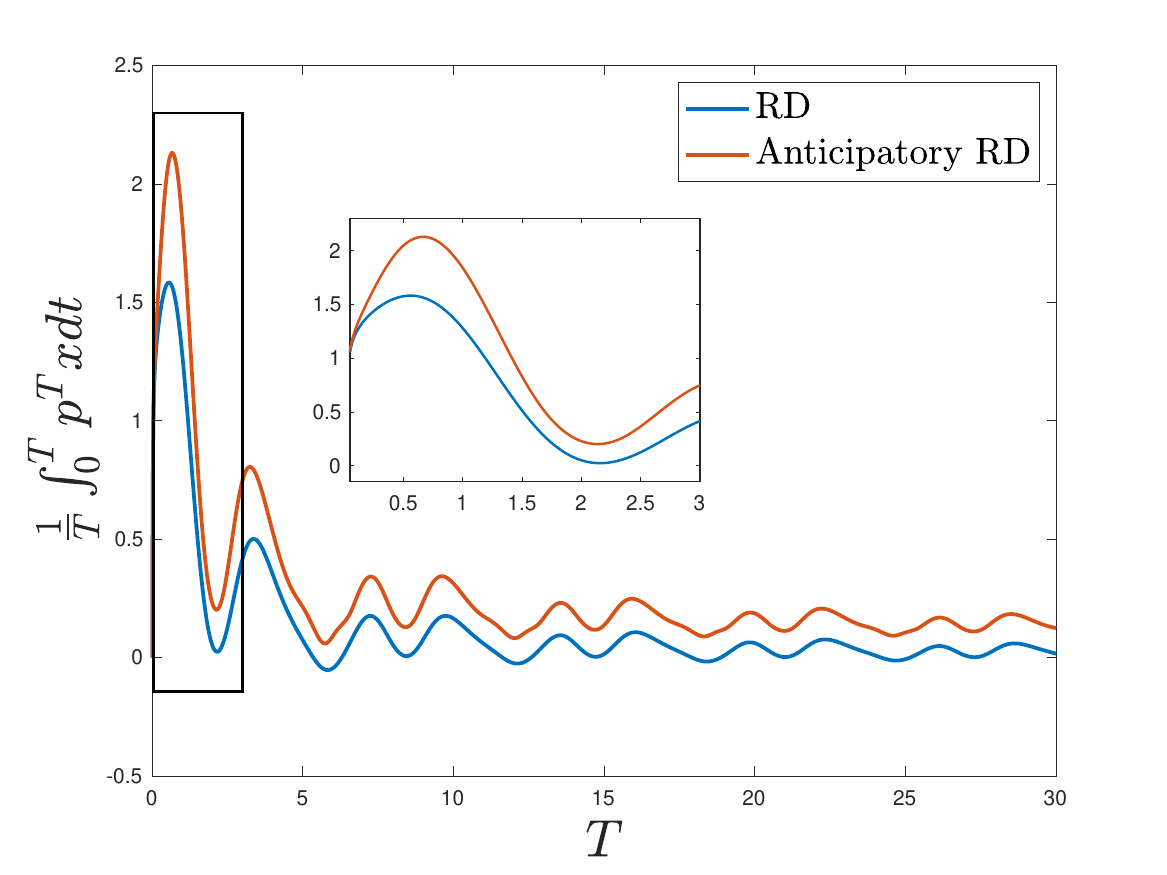}
    \caption{Performance comparison between anticipatory RD and standard RD for the payoff environment in Example~\ref{ex:global_performance_comparison}.}
    \label{fig:global_comparison_RD_anti_RD}
\end{figure}

\section{Conclusion and future work}

This paper studied performance dominance among learning dynamics from a control--theoretic perspective, with particular emphasis on the existence of ``free--lunch'' phenomena among no--regret algorithms. We showed that regret guarantees alone are insufficient to assess learning performance and introduced
a systematic framework for comparing learning rules through cumulative reward. By modeling payoff-based higher-order replicator dynamics as a cascade interconnection between a linear time--invariant (LTI) system and the softmax mapping, we linked the asymptotic average reward of the learning dynamics
directly to the frequency response of the underlying LTI transfer function.

Within this framework, we established that oracle RD uniformly dominates standard replicator dynamics for all payoff trajectories. We further showed that predictive exponential RD equipped with a first--order low--pass predictor uniformly dominates its standard variant. By formulating performance comparison as a passivity problem, we demonstrated that predictive RD  with any passive, asymptotically stable predictor locally uniformly dominates standard RD. Finally, by casting the comparison between anticipatory and standard replicator dynamics as an optimal--control problem, we showed that the minimal performance gap is zero. This result implies global uniform dominance of anticipatory replicator dynamics for arbitrary payoff environments and answers the central question of this paper: a learner can indeed regret not using a no--regret algorithm.

Several directions for future research naturally follow. First, it remains open to establish a global performance dominance result for predictive replicator dynamics with general passive  predictors. Second, an important and largely unexplored question is whether ``free-lunch'' phenomena persist in strategic
multi-agent settings. In particular, when agents interact through a game and employ no--regret learning rules, can an agent improve its long-run performance by selecting a specific no--regret algorithm given knowledge of the opponents' learning dynamics across general game environments?
Addressing these questions would further clarify when and how a learner may regret using one no--regret algorithm over another in strategic settings.

\section{references}

\bibliographystyle{ieeetr}

\bibliography{refs.bib}

\end{document}